\documentclass[12pt]{article}
\usepackage[margin=1.25in]{geometry}

\usepackage{amssymb}
\usepackage{amsmath}
\usepackage{graphicx}
\usepackage{wrapfig}
\usepackage{authblk}
\usepackage[all]{xy}
\usepackage{bbm}
\usepackage{float}
\usepackage{tikz}
\usetikzlibrary{arrows,calc,shapes,decorations.pathreplacing}
\usepackage{tikz}
\usetikzlibrary{arrows,calc,shapes,decorations.pathreplacing}
\usepackage{sidecap}
\usepackage{hyperref}
\hypersetup{pdftex,colorlinks=true,allcolors=blue}
\usepackage{hypcap}
\usepackage[font={small,it}]{caption}
\usepackage{pgfplots}
\usetikzlibrary{patterns}
\usepackage{enumerate}
\usepackage{subcaption} 

\newenvironment{customlegend}[1][]{%
    \begingroup
    \csname pgfplots@init@cleared@structures\endcsname
    \pgfplotsset{#1}%
}{%
    \csname pgfplots@createlegend\endcsname
    \endgroup
}%

\def\addlegendimage{\csname pgfplots@addlegendimage\endcsname}

\newtheorem{theorem}{Theorem}[section]

\newtheorem{proposition}[theorem]{Proposition}

\newtheorem{definition}[theorem]{Definition}
\newtheorem{remark}[theorem]{Remark}

\newenvironment{proof}[1][Proof]{\begin{trivlist}
\item[\hskip \labelsep {\bfseries #1}]}{\end{trivlist}}

\newcommand{\argmin}{\operatornamewithlimits{argmin}}
\newcommand{\qed}{\nobreak \ifvmode \relax \else
      \ifdim\lastskip<1.5em \hskip-\lastskip
      \hskip1.5em plus0em minus0.5em \fi \nobreak
      \vrule height0.75em width0.5em depth0.25em\fi}
      
\begin{document}

\title{A strategic timing of arrivals to a linear slowdown processor sharing system\footnote{To appear in the European Journal of Operational Research}}
\author[1]{Liron Ravner}
\author[1]{Moshe Haviv}
\author[2]{Hai L. Vu}
\affil[1]{\small{Department of Statistics and the Federmann Center for the Study of Rationality, The Hebrew University of Jerusalem, Israel}}
\affil[2]{\small{Intelligent Transport Systems Lab, Swinburne University of Technology, Australia}}

\date{\today}
\maketitle

%%%%%%%%%%%%%%%%%%%%%%%%%%%%%
%%%%%%%%%%%%%%%%%%%%%%%%%%%%%
\begin{abstract}
We consider a discrete population of users with homogeneous service demand who need to decide when to arrive to a system in which the service rate deteriorates linearly with the number of users in the system. The users have heterogeneous desired departure times from the system, and their goal is to minimize a weighted sum of the travel time and square deviation from the desired departure times. Users join the system sequentially, according to the order of their desired departure times. We model this scenario as a non-cooperative game in which each user selects his actual arrival time. We present explicit equilibria solutions for a two-user example, namely the Subgame Perfect and Cournot Nash equilibria and show that multiple equilibria may exist. We further explain why a general solution for any number of users is computationally challenging. The difficulty lies in the fact that the objective functions are piecewise-convex, i.e., non-smooth and non-convex. As a result, the minimization of the costs relies on checking all arrival and departure order permutations, which is exponentially large with respect to the population size. Instead we propose an iterated best-response algorithm which can be efficiently studied numerically. Finally, we compare the equilibrium arrival profiles to a socially optimal solution and discuss the implications.
\end{abstract}

%%%%%%%%%%%%%%%%%%%%%%%%%%%%%
\section{Introduction}\label{sec:intro}
The strategic timing of arrivals to congested systems is relevant for various applications such as traffic, queueing and communication networks. We study a non-cooperative game in which atomic users need to time their arrival to a deterministic processor sharing system with linear slowdown. This may be the case on a ring road around a business district in which the density at any point on the road affects all of the road, and therefore arriving users can cause a slowdown even for users who arrived before them. Throughout the paper we refer to the model as a traffic network where users travel along a route at varying speeds. Nevertheless, our aim here is to provide a general analysis of the strategic arrival times to such a processor sharing system, and the results are not limited to the specific traffic application. The linear slowdown dynamic can be seen as a discrete variation of the Greenshield's fluid model (see for example \cite{MH1984}). This work complements \cite{RN2015} where the socially optimal arrival schedule of users to the same system was analysed. In this paper the choice of arrival times is made by the users themselves sequentially, according to the their desired departure times. Note that while all users are served simultaneously, the model presented here still maintains the First-In-First-Out property, and thus in the sequential game users leave the system in the same order they arrived. 

The study of departure time choice to a congested bottleneck goes back to Vickery \cite{V1969}, where a fluid queue dynamic was assumed. The research of fluid bottleneck models has evolved greatly since then, and we refer the reader to Arnott et al. \cite{ADL1993} and de Dios Ort{\'u}zar and Willumsen \cite{book_OW2011} and references therein. Otsubo and Rapoport analysed a non-fluid (atomic user) game with discrete arrival instances in \cite{OR2008}. An arrival time and route choice (dynamic user equilibrium) game for a route with linear slowdown was analysed using a mean field approach by Mahmassani and Herman in \cite{MH1984}.

A queueing theory approach to the strategic timinig of arrivals to a congested stochastic queue was developed by Glazer and Hassin in \cite{GH1983}. They introduced a game in which a discrete population of users, of a Poisson distributed size, choose arrival times to a single server exponential queue with the goal of minimizing waiting times. This led to another branch of research that relies on the stochastic properties of the queues, rather than fluid dynamics. Examples with a discrete deterministic population, as we assume in this work, are the works of Juneja and Shimkin in \cite{JS2013} (tardniess costs), Ravner \cite{R2014} (order penalties), and Haviv and Ravner \cite{HR2015} (loss system). All of the above assume random memoryless service times and a first come first served regime. An arrival time game to a processor sharing queue was studied by Raheja and Juneja in \cite{JR2014}, using a fluid approximation. In the queueing context our work is the first to analyse an arrival time game with a discrete population queue with heterogeneous users and deterministic service times.

%%%%%%%%%%%%%%%%%%%%%%%%%%%%%
\section{Traffic model}\label{sec:traffic_model}
In this section we introduce the model and show that despite its seeming simplicity, it in fact yields a complex arrival-departure dynamic. Suppose a set of atomic users $\mathcal{N}:=\{1,\ldots,N\}$ need to travel on a single route of length $1$. We define the travel speed on a segment at time $t$ as:
\begin{equation}\label{eq:speed}
v(t):=\beta-\alpha(q(t)-1),
\end{equation}
where $q(t)$ is the number of users on the segment at time $t$, $\beta>0$ is the free flow speed of a single user travelling alone, and $\alpha\geq 0$ is the slowdown parameter. We assume that $\beta-\alpha(N-1)>0$, in other words, this means that the travel speed is positive even if all users travel at the same time. Note that for $\alpha\geq\frac{\beta}{2}$, the only possible case is $N=2$. In Figure \ref{fig:service_dynamics} the system dynamics are illustrated as a function of the number of concurrent users. Observe that while travel speed decreases by definition, the overall service rate is non-monotone and concave. In particular, it is initially increasing, has a maximal throughput at level $q=\frac{\beta+\alpha}{2\alpha}$ and then decreases to almost zero when the system is very busy.

\begin{figure}[H]
\begin{subfigure}{\textwidth}
\centering
\begin{tikzpicture}[xscale=0.12,yscale=2.7]
  \def\xmin{0}
  \def\xmax{101}
  \def\ymin{0}
  \def\ymax{1.03}
  
    \draw[->] (\xmin,\ymin) -- (\xmax,\ymin) node[right] {$q$} ;
    \draw[->] (\xmin,\ymin) -- (\xmin,\ymax);% node[above] {$\beta-\alpha(N-1)$} ;    
    \foreach \x in {0,10,20,30,40,50,60,70,80,90,100}
    \node at (\x,\ymin) [below] {\x};
    \foreach \y in {0,0.2,0.4,0.6,0.8,1}
    \node at (\xmin,\y) [left] {\y};
  
  \draw[red, dotted, thick,domain=0:100]  plot (\x, {1-0.01*\x});
  \draw[] (38,0.75) node[right,red] {$\beta-\alpha(q-1)$};  
\end{tikzpicture}
\caption{Travel speed}
\end{subfigure}

\begin{subfigure}{\textwidth}
\centering
\begin{tikzpicture}[xscale=0.12,yscale=0.108]
  \def\xmin{0}
  \def\xmax{101}
  \def\ymin{0}
  \def\ymax{26}
  
    \draw[->] (\xmin,\ymin) -- (\xmax,\ymin) node[right] {$q$} ;
    \draw[->] (\xmin,\ymin) -- (\xmin,\ymax);% node[above] {$\beta-\alpha(N-1)$} ;
        \foreach \x in {0,10,20,30,40,50,60,70,80,90,100}
    \node at (\x,\ymin) [below] {\x};
    \foreach \y in {0,5,10,15,20,25}
    \node at (\xmin,\y) [left] {\y};
  
  \draw[blue, densely dotted, thick,domain=0:100]  plot (\x, {\x*(1-0.01*\x)});
  \draw[] (38,17) node[right,blue] {$q(\beta-\alpha(q-1))$};
  
\end{tikzpicture}
\caption{Aggregate throughput}
\end{subfigure}
\caption{The individual travel speed and aggregate throughput as the number of users in the system increases. Parameter values: $\beta=1$ and $\alpha=0.01$}
\label{fig:service_dynamics}
\end{figure}
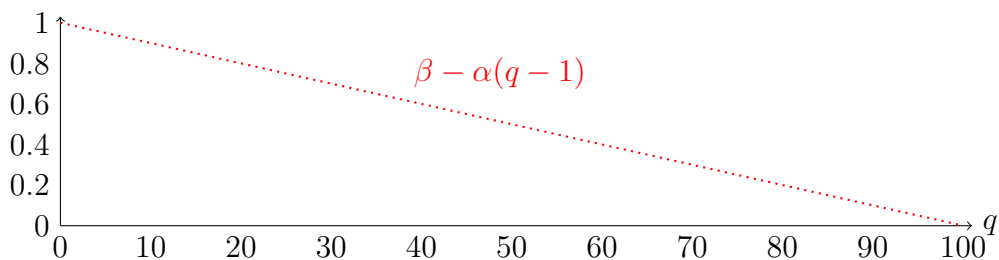
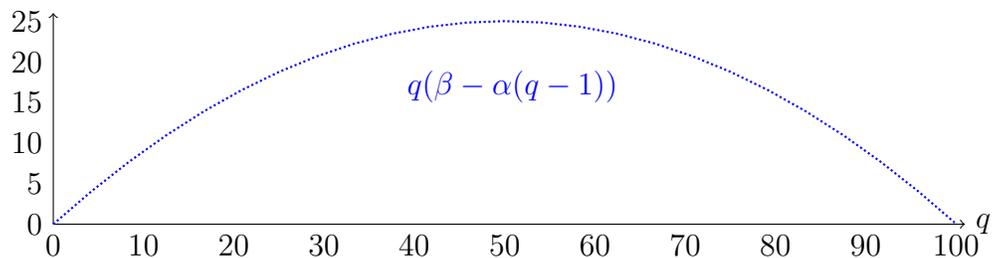

Every user $i\in\mathcal{N}$ has a desired departure time from the system, denoted by $d_i^*$. Without loss of generality we assume that the desired departure times are ordered: $d_i^*\leq d_j^*,\ \forall i<j$. The action of user $i$ is choosing an arrival time $a_i\in\mathbbm{R}$. Denote the arrival and departure vectors of all users by $\mathbf{a}:=(a_1,\ldots,a_N)$ and $\mathbf{d}:=(d_1,\ldots,d_N)$, respectively. The cost incurred by user $i$ is
\begin{equation}\label{eq:costUserShort}
c_i(\mathbf{a})= (d_i-d_i^*)^2 \,+ \gamma \,  (d_i-a_i).
\end{equation}
This cost function is a combination of a quadratic penalty for any deviation from the desired departure time, be it early or late, and a linear penalty for the total travel time. We focus on this cost function for the sake of a clear presentation, but all of our analysis can be applied to any convex function (of the deviation and travel time terms) in a straightforward manner. We further discuss on how this generalization can be made in the concluding remarks in Section \ref{sec:conclusion}. The minimal cost of any user is $\frac{\gamma}{\beta}$, and can only be obtained by travelling alone at free flow speed and leaving at exactly the desired time, $d_i^*$ for user $i$. 

The effective departure times of users are determined by $\mathbf{a}$ and the travel dynamics defined in \eqref{eq:speed}:
\begin{equation}
\label{eq:dynamics}
1 = \int_{a_i}^{d_i} v\left(t\right) dt,
\qquad  i\in\mathcal{N},
\end{equation}
where $q(t) = \sum_{i\in\mathcal{N}} {\mathbbm 1}_{\{t \in [a_i,d_i]\}}$. Using \eqref{eq:speed} we get a set of $N$ equations for $\mathbf{d}$,
\[
1 = (d_i-a_i)(\beta+\alpha)-\alpha\int_{a_i}^{d_i} \sum_{j\in\mathcal{N}}  {\mathbbm 1}\{t \in [a_j,d_j]\}dt,
\qquad
i\in\mathcal{N}.
\]
These $N$ equations can be treated as equations for the unknowns ${\mathbf d}$, given ${\mathbf a}$ or vice-versa. In \cite{RN2015} it was shown that the departure dynamics for an ordered vector $\mathbf{a}$ are given by
\begin{equation*}\label{eq:daMatrix}
D\, {\mathbf d} - A \, {\mathbf a} = {\mathbf 1},
\end{equation*}
where $A\in\mathbbm{R}^N$ and $D\in\mathbbm{R}^N$ are defined as follows:
\[
A_{ij}:=
\left\{
	\begin{array}{ll}
		\beta-\alpha(i-h_i) \mbox{, } &  i=j \\
		-\alpha \mbox{, } &  i+1\leq j\leq k_i  \\
		0\mbox{, } &  o.w.
	\end{array}
\right. , \ and
\]
\[
D_{ij}:=
\left\{
	\begin{array}{ll}
		\beta-\alpha(k_i-i) \mbox{, } &  i=j \\
		-\alpha \mbox{, } &  h_i\leq j\leq i-1  \\
		0\mbox{, } &  o.w.
	\end{array}
\right. ,
\]
with $k_i:=\max \big\{ k \in\mathcal{N} ~:~ a_k \le d_i \big\}$ and $h_i:=\min 
\big\{ h \in\mathcal{N} ~:~ d_h \ge a_i\big\}$.
A direct result of this is the recursive formula
\begin{equation}\label{eq:d_recursive}
d_{i}=
\frac{
1 + \left(\beta - \alpha(i-h_{i})\right)a_i + \alpha \left(
\sum_{j=h_{i}}^{i-1} d_j - \sum_{j=i+1}^{k_{i}} a_j
\right)
}{\beta-\alpha(k_{i}-i)},\ i\in\mathcal{N}.
\end{equation}
Using an iterative algorithm we can compute the unique $\mathbf{d}$ for any given $\mathbf{a}$ (or vice versa) with at most $2N$ computations. At this point it is important to observe that the vector $\mathbf{k}:=(k_1,\ldots,k_N)$ defines the combined order permutation of all arrivals and departures. For example, if $N=3$, the profile $(a_1<a_2<d_1<a_3<d_2<d_3)$ corresponds to $\mathbf{k}:=(2,3,3)$. In this example there is an overlay between users $1$ and $2$ during the interval $[a_2,d_1]$, and between users $2$ and $3$ during the interval $[a_3,d_2]$. Therefore, given $\mathbf{k}$ we know the exact number of users and their travel speed at any point in time. Furthermore, as long as there is no change in the permutation, the departure times are continuous with the arrival times with a known linear coefficient, as shown in Figure \ref{fig:departure_arrival}. Note that \eqref{eq:d_recursive} simply solves the traffic dynamics without any consideration of the cost function \eqref{eq:costUserShort} that users wish to minimize.

If we denote the set of all possible arrival-departure permutations by
\[
\mathcal{K}:=\left\lbrace\mathbf{k}\in{\cal N}^N \,:\,  k_N=N,\, k_i \le k_{j} \ \forall i \le j\right\rbrace,
\]
then $|{\cal K}| = \frac{{2N \choose N}}{N+1}$. This follows by observing that the elements of ${\cal K}$ correspond uniquely to lattice paths in the $N \times N$ grid from bottom-left to top-right with up and right movements without crossing the diagonal. The number of such elements is the Catalan number (see \cite{RN2015} for more details).

The relation between arrival and departure times, defined in \eqref{eq:dynamics}, is in fact a set of piecewise-linear equations. This is illustrated in Figure \ref{fig:departure_arrival} by changing the arrival time of a single user while keeping all others fixed. 

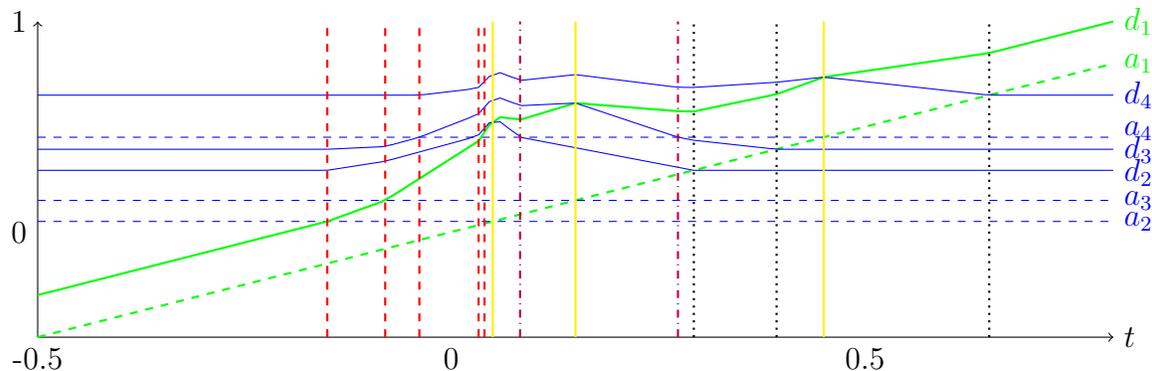
\begin{figure}[H]
\begin{tikzpicture}[xscale=11,yscale=2.8]
  \def\xmin{-0.5}
  \def\xmax{0.8}
  \def\ymin{-0.5}
  \def\ymax{1}
    \draw[->] (\xmin,\ymin) -- (\xmax,\ymin) node[right] {$t$} ;
    \draw[->] (\xmin,\ymin) -- (\xmin,\ymax); %node[above] {Arrival and departure times} 
    \foreach \x in {-0.5,0,0.5}
    \node at (\x,\ymin) [below] {\x};
    \foreach \y in {0,1}
    \node at (\xmin,\y) [left] {\y};
    
    \draw[smooth,dashed,green,thick] (-0.5,-0.5)--	(-0.1621,-0.1621)--	(-0.1491,-0.1491)--	(-0.1361,-0.1361)--	(-0.1231,-0.1231)--	(-0.1101,-0.1101)--	(-0.0971,-0.0971)--	(-0.0841,-0.0841)--	(-0.0711,-0.0711)--	(-0.0581,-0.0581)--	(-0.0451,-0.0451)--	(-0.0321,-0.0321)--	(-0.0191,-0.0191)--	(-0.0061,-0.0061)--	(0.0069,0.0069)--	(0.0199,0.0199)--	(0.0329,0.0329)--	(0.0459,0.0459)--	(0.0589,0.0589)--	(0.0719,0.0719)--	(0.0849,0.0849)--	(0.0979,0.0979)--	(0.1109,0.1109)--	(0.1239,0.1239)--	(0.1369,0.1369)--	(0.1499,0.1499)--	(0.1629,0.1629)--	(0.1759,0.1759)--	(0.1889,0.1889)--	(0.2019,0.2019)--	(0.2149,0.2149)--	(0.2279,0.2279)--	(0.2409,0.2409)--	(0.2539,0.2539)--	(0.2669,0.2669)--	(0.2799,0.2799)--	(0.2929,0.2929)--	(0.3059,0.3059)--	(0.3189,0.3189)--	(0.3319,0.3319)--	(0.3449,0.3449)--	(0.3579,0.3579)--	(0.3709,0.3709)--	(0.3839,0.3839)--	(0.3969,0.3969)--	(0.4099,0.4099)--	(0.4229,0.4229)--	(0.4359,0.4359)--	(0.4489,0.4489)--	(0.4619,0.4619)--	(0.4749,0.4749)--	(0.4879,0.4879)--	(0.5009,0.5009)--	(0.5139,0.5139)--	(0.5269,0.5269)--	(0.5399,0.5399)--	(0.5529,0.5529)--	(0.5659,0.5659)--	(0.5789,0.5789)--	(0.5919,0.5919)--	(0.6049,0.6049)--	(0.6179,0.6179)--	(0.6309,0.6309)--	(0.6439,0.6439)--	(0.6569,0.6569)--	(0.6699,0.6699)--	(0.6829,0.6829)--	(0.6959,0.6959)--	(0.7089,0.7089)--	(0.7219,0.7219)--	(0.7349,0.7349)--	(0.7479,0.7479)--	(0.7609,0.7609)--	(0.7739,0.7739)--	(0.7869,0.7869)--	(0.7999,0.7999);
    
     \draw[smooth,dashed,blue] (-0.5,0.05) -- (0.8,0.05);
     
     \draw[smooth,dashed,blue] (-0.5,0.15) -- (0.8,0.15);
     
     \draw[smooth,dashed,blue] (-0.5,0.45) -- (0.8,0.45);

    \draw[smooth,green,thick] (-0.5,-0.3)--	(-0.1621,0.0379)--	(-0.1491,0.05128571)--	(-0.1361,0.06985714)--	(-0.1231,0.08842857)--	(-0.1101,0.107)--	(-0.0971,0.12557143)--	(-0.0841,0.14414286)--	(-0.0711,0.17225)--	(-0.0581,0.20475)--	(-0.0451,0.23725)--	(-0.0321,0.26975)--	(-0.0191,0.30225)--	(-0.0061,0.33475)--	(0.0069,0.36725)--	(0.0199,0.39975)--	(0.0329,0.43225)--	(0.0459,0.509)--	(0.0589,0.54555)--	(0.0719,0.53905)--	(0.0849,0.53519375)--	(0.0979,0.55063125)--	(0.1109,0.56606875)--	(0.1239,0.58150625)--	(0.1369,0.59694375)--	(0.1499,0.61238125)--	(0.1629,0.60846875)--	(0.1759,0.60440625)--	(0.1889,0.60034375)--	(0.2019,0.59628125)--	(0.2149,0.59221875)--	(0.2279,0.58815625)--	(0.2409,0.58409375)--	(0.2539,0.58003125)--	(0.2669,0.57596875)--	(0.2799,0.57337449)--	(0.2929,0.57248397)--	(0.3059,0.58309621)--	(0.3189,0.59370845)--	(0.3319,0.6043207)--	(0.3449,0.61493294)--	(0.3579,0.62554519)--	(0.3709,0.63615743)--	(0.3839,0.64676968)--	(0.3969,0.65985714)--	(0.4099,0.67842857)--	(0.4229,0.697)--	(0.4359,0.71557143)--	(0.4489,0.73414286)--	(0.4619,0.74251429)--	(0.4749,0.74994286)--	(0.4879,0.75737143)--	(0.5009,0.7648)--	(0.5139,0.77222857)--	(0.5269,0.77965714)--	(0.5399,0.78708571)--	(0.5529,0.79451429)--	(0.5659,0.80194286)--	(0.5789,0.80937143)--	(0.5919,0.8168)--	(0.6049,0.82422857)--	(0.6179,0.83165714)--	(0.6309,0.83908571)--	(0.6439,0.84651429)--	(0.6569,0.8569)--	(0.6699,0.8699)--	(0.6829,0.8829)--	(0.6959,0.8959)--	(0.7089,0.9089)--	(0.7219,0.9219)--	(0.7349,0.9349)--	(0.7479,0.9479)--	(0.7609,0.9609)--	(0.7739,0.9739)--	(0.7869,0.9869)--	(0.7999,0.9999);

\draw[smooth,blue] (-0.5,0.2928571)--	(-0.1621,0.2928571)--	(-0.1491,0.2934082)--	(-0.1361,0.3013673)--	(-0.1231,0.3093265)--	(-0.1101,0.3172857)--	(-0.0971,0.3252449)--	(-0.0841,0.3332041)--	(-0.0711,0.34525)--	(-0.0581,0.3591786)--	(-0.0451,0.3731071)--	(-0.0321,0.3870357)--	(-0.0191,0.4009643)--	(-0.0061,0.4148929)--	(0.0069,0.4288214)--	(0.0199,0.44275)--	(0.0329,0.4616875)--	(0.0459,0.51925)--	(0.0589,0.5233)--	(0.0719,0.4843)--	(0.0849,0.448825)--	(0.0979,0.439075)--	(0.1109,0.429325)--	(0.1239,0.419575)--	(0.1369,0.409825)--	(0.1499,0.400075)--	(0.1629,0.390325)--	(0.1759,0.380575)--	(0.1889,0.370825)--	(0.2019,0.361075)--	(0.2149,0.351325)--	(0.2279,0.341575)--	(0.2409,0.331825)--	(0.2539,0.322075)--	(0.2669,0.312325)--	(0.2799,0.302575)--	(0.2929,0.2928571)--	(0.3059,0.2928571)--	(0.3189,0.2928571)--	(0.3319,0.2928571)--	(0.3449,0.2928571)--	(0.3579,0.2928571)--	(0.3709,0.2928571)--	(0.3839,0.2928571)--	(0.3969,0.2928571)--	(0.4099,0.2928571)--	(0.4229,0.2928571)--	(0.4359,0.2928571)--	(0.4489,0.2928571)--	(0.4619,0.2928571)--	(0.4749,0.2928571)--	(0.4879,0.2928571)--	(0.5009,0.2928571)--	(0.5139,0.2928571)--	(0.5269,0.2928571)--	(0.5399,0.2928571)--	(0.5529,0.2928571)--	(0.5659,0.2928571)--	(0.5789,0.2928571)--	(0.5919,0.2928571)--	(0.6049,0.2928571)--	(0.6179,0.2928571)--	(0.6309,0.2928571)--	(0.6439,0.2928571)--	(0.6569,0.2928571)--	(0.6699,0.2928571)--	(0.6829,0.2928571)--	(0.6959,0.2928571)--	(0.7089,0.2928571)--	(0.7219,0.2928571)--	(0.7349,0.2928571)--	(0.7479,0.2928571)--	(0.7609,0.2928571)--	(0.7739,0.2928571)--	(0.7869,0.2928571)--	(0.7999,0.2928571);

\draw[smooth,blue] (-0.5,0.3928571)--	(-0.1621,0.3928571)--	(-0.1491,0.3930224)--	(-0.1361,0.3954102)--	(-0.1231,0.397798)--	(-0.1101,0.4001857)--	(-0.0971,0.4025735)--	(-0.0841,0.4049612)--	(-0.0711,0.41525)--	(-0.0581,0.4291786)--	(-0.0451,0.4431071)--	(-0.0321,0.460051)--	(-0.0191,0.479949)--	(-0.0061,0.4998469)--	(0.0069,0.5197449)--	(0.0199,0.5396429)--	(0.0329,0.5616875)--	(0.0459,0.61925)--	(0.0589,0.63665)--	(0.0719,0.61715)--	(0.0849,0.6002937)--	(0.0979,0.6027312)--	(0.1109,0.6051687)--	(0.1239,0.6076062)--	(0.1369,0.6100437)--	(0.1499,0.6124813)--	(0.1629,0.5955687)--	(0.1759,0.5785062)--	(0.1889,0.5614437)--	(0.2019,0.5443812)--	(0.2149,0.5273188)--	(0.2279,0.5102562)--	(0.2409,0.4931938)--	(0.2539,0.4761313)--	(0.2669,0.4590687)--	(0.2799,0.4454321)--	(0.2929,0.4356959)--	(0.3059,0.4301245)--	(0.3189,0.4245531)--	(0.3319,0.4189816)--	(0.3449,0.4134102)--	(0.3579,0.4078388)--	(0.3709,0.4022673)--	(0.3839,0.3966959)--	(0.3969,0.3928571)--	(0.4099,0.3928571)--	(0.4229,0.3928571)--	(0.4359,0.3928571)--	(0.4489,0.3928571)--	(0.4619,0.3928571)--	(0.4749,0.3928571)--	(0.4879,0.3928571)--	(0.5009,0.3928571)--	(0.5139,0.3928571)--	(0.5269,0.3928571)--	(0.5399,0.3928571)--	(0.5529,0.3928571)--	(0.5659,0.3928571)--	(0.5789,0.3928571)--	(0.5919,0.3928571)--	(0.6049,0.3928571)--	(0.6179,0.3928571)--	(0.6309,0.3928571)--	(0.6439,0.3928571)--	(0.6569,0.3928571)--	(0.6699,0.3928571)--	(0.6829,0.3928571)--	(0.6959,0.3928571)--	(0.7089,0.3928571)--	(0.7219,0.3928571)--	(0.7349,0.3928571)--	(0.7479,0.3928571)--	(0.7609,0.3928571)--	(0.7739,0.3928571)--	(0.7869,0.3928571)--	(0.7999,0.3928571);

\draw[smooth,blue] (-0.5,0.65)--	(-0.1621,0.65)--	(-0.1491,0.65)--	(-0.1361,0.65)--	(-0.1231,0.65)--	(-0.1101,0.65)--	(-0.0971,0.65)--	(-0.0841,0.65)--	(-0.0711,0.65)--	(-0.0581,0.65)--	(-0.0451,0.65)--	(-0.0321,0.6530153)--	(-0.0191,0.6589847)--	(-0.0061,0.6649541)--	(0.0069,0.6709235)--	(0.0199,0.6768929)--	(0.0329,0.6870125)--	(0.0459,0.73925)--	(0.0589,0.75665)--	(0.0719,0.73715)--	(0.0849,0.7206462)--	(0.0979,0.7260087)--	(0.1109,0.7313712)--	(0.1239,0.7367337)--	(0.1369,0.7420963)--	(0.1499,0.7474588)--	(0.1629,0.7412112)--	(0.1759,0.7348737)--	(0.1889,0.7285362)--	(0.2019,0.7221988)--	(0.2149,0.7158613)--	(0.2279,0.7095238)--	(0.2409,0.7031863)--	(0.2539,0.6968488)--	(0.2669,0.6905112)--	(0.2799,0.6870123)--	(0.2929,0.6867452)--	(0.3059,0.6899289)--	(0.3189,0.6931125)--	(0.3319,0.6962962)--	(0.3449,0.6994799)--	(0.3579,0.7026636)--	(0.3709,0.7058472)--	(0.3839,0.7090309)--	(0.3969,0.7129571)--	(0.4099,0.7185286)--	(0.4229,0.7241)--	(0.4359,0.7296714)--	(0.4489,0.7352429)--	(0.4619,0.7306143)--	(0.4749,0.7250429)--	(0.4879,0.7194714)--	(0.5009,0.7139)--	(0.5139,0.7083286)--	(0.5269,0.7027571)--	(0.5399,0.6971857)--	(0.5529,0.6916143)--	(0.5659,0.6860429)--	(0.5789,0.6804714)--	(0.5919,0.6749)--	(0.6049,0.6693286)--	(0.6179,0.6637571)--	(0.6309,0.6581857)--	(0.6439,0.6526143)--	(0.6569,0.65)--	(0.6699,0.65)--	(0.6829,0.65)--	(0.6959,0.65)--	(0.7089,0.65)--	(0.7219,0.65)--	(0.7349,0.65)--	(0.7479,0.65)--	(0.7609,0.65)--	(0.7739,0.65)--	(0.7869,0.65)--	(0.7999,0.65);
     
     \draw[smooth,yellow,thick] (0.05,-0.5) -- (0.05,1);
     \draw[smooth,yellow,thick] (0.15,-0.5) -- (0.15,1);
     \draw[smooth,yellow,thick] (0.45,-0.5) -- (0.45,1);
     
     \draw[smooth,dotted,black,thick] (0.293,-0.5) -- (0.293,1);
     \draw[smooth,dotted,black,thick] (0.393,-0.5) -- (0.393,1);
     \draw[smooth,dotted,black,thick] (0.65,-0.5) -- (0.65,1);
     
     \draw[smooth,dashed,red,thick] (-0.15,-0.5) -- (-0.15,1);
     \draw[smooth, dashed,red,thick] (-0.08,-0.5) -- (-0.08,1);
     \draw[smooth, dashed,red,thick] (-0.0387,-0.5) -- (-0.0387,1);
	 \draw[smooth, dashed,red,thick] (0.0329,-0.5) -- (0.0329,1);     
     \draw[smooth, dashed,red,thick] (0.04,-0.5) -- (0.04,1);     
     
     \draw[smooth, dashdotted,purple,thick] (0.083,-0.5) -- (0.083,1);
     \draw[smooth, dashdotted,purple,thick] (0.2738,-0.5) -- (0.2738,1);
     
     \draw[] (0.8,0.8) node[right,green] {$a_1$};
     \draw[] (0.8,1) node[right,green] {$d_1$};
     \draw[] (0.8,0.05) node[right,blue] {$a_2$};
     \draw[] (0.8,0.283) node[right,blue] {$d_2$};
     \draw[] (0.8,0.15) node[right,blue] {$a_3$};
     \draw[] (0.8,0.393) node[right,blue] {$d_3$};
     \draw[] (0.8,0.48) node[right,blue] {$a_4$};
     \draw[] (0.8,0.65) node[right,blue] {$d_4$};
     
%     \draw[smooth,dotted,yellow,thick] (0.05,-0.5) -- (0.05,1);
%     \draw[smooth,dotted,yellow,thick] (0.15,-0.5) -- (0.15,1);
%     \draw[smooth,dotted,yellow,thick] (0.45,-0.5) -- (0.45,1);
%     
%     \draw[smooth,dotted,black,thick] (0.293,-0.5) -- (0.293,1);
%     \draw[smooth,dotted,black,thick] (0.393,-0.5) -- (0.393,1);
%     \draw[smooth,dotted,black,thick] (0.65,-0.5) -- (0.65,1);
%     
%     \draw[smooth,dotted,red,thick] (-0.15,-0.5) -- (-0.15,1);
%     \draw[smooth,dotted,red,thick] (-0.08,-0.5) -- (-0.08,1);
%     \draw[smooth,dotted,red,thick] (-0.0387,-0.5) -- (-0.0387,1);
%     \draw[smooth,dotted,red,thick] (0.04,-0.5) -- (0.04,1);
%     \draw[smooth,dotted,red,thick] (0.2738,-0.5) -- (0.2738,1);
%     
%     \draw[] (0.8,0.8) node[right,green] {$a_1$};
%     \draw[] (0.8,1) node[right,green] {$d_1$};
%     \draw[] (0.8,0.05) node[right,blue] {$a_2$};
%     \draw[] (0.8,0.283) node[right,blue] {$d_2$};
%     \draw[] (0.8,0.15) node[right,blue] {$a_3$};
%     \draw[] (0.8,0.393) node[right,blue] {$d_3$};
%     \draw[] (0.8,0.48) node[right,blue] {$a_4$};
%     \draw[] (0.8,0.65) node[right,blue] {$d_4$};
\end{tikzpicture}
\caption{\textbf{Arrival-departure dynamics:} An illustration of the effect of changing the arrival time of one user $1$ (green user) while keeping the arrival times of the other three users $(2,3,4)$ (blue) constant. The horizontal axis is the arrival time of the green user. The vertical axis shows the arrival and departure times of all other users: Arrival times of blue users are the blue dashed horizontal lines at $(0.05, 0.15, 0.45)$ and the arrival time of the green user varies and is thus represented by the dashed diagonal line. Departure times are the solid lines. The numerous vertical dotted lines illustrate ``break-points". In between these lines, the effect of changing the green user's arrival on all users is affine. The red vertical line marks a point where a departure overtakes an arrival, the yellow vertical line marks a point where an arrival overtakes an arrival and the black marks a point where an arrival overtakes a departure. The parameter values of the example are: $\alpha=1.5$ and $\beta=5$.}
\label{fig:departure_arrival}
\end{figure}

The piecwise linear relation between arrival and departure times implies that the cost of any single user is not convex with respect to his own arrival time, even though the cost function, defined in \eqref{eq:costUserShort}, has a convex form. This piecewise-convex behaviour is illustrated for the same numerical example in Figure \ref{fig:cost_arrival}. We can see that the cost of any user is not convex, monotone, or differentiable with respect to a change in arrival time of any user (including himself). Clearly, this complex cost structure leads to difficulties both in finding a socially optimal arrival schedule and a Nash equilibrium schedule.

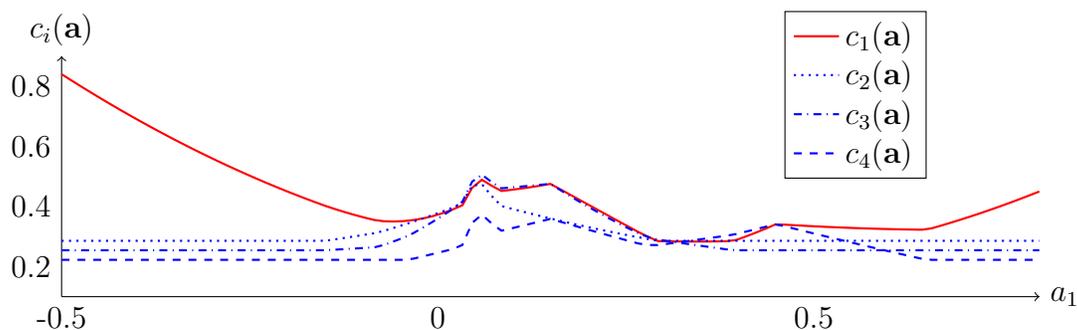
\begin{figure}[H]
\centering
\begin{tikzpicture}[xscale=10,yscale=4]
  \def\xmin{-0.5}
  \def\xmax{0.8}
  \def\ymin{0.1}
  \def\ymax{0.9}
    \draw[->] (\xmin,\ymin) -- (\xmax,\ymin) node[right] {$a_1$} ;
    \draw[->] (\xmin,\ymin) -- (\xmin,\ymax) node[above] {$c_i(\mathbf{a})$} ;
    \foreach \x in {-0.5,0,0.5}
    \node at (\x,\ymin) [below] {\x};
    \foreach \y in {0.2,0.4,0.6,0.8}
    \node at (\xmin,\y) [left] {\y};
    
    \draw[smooth,red, thick] (-0.5,0.84)--	(-0.4872,0.8196838)--	(-0.4742,0.7993856)--	(-0.4612,0.7794254)--	(-0.4482,0.7598032)--	(-0.4352,0.740519)--	(-0.4222,0.7215728)--	(-0.4092,0.7029646)--	(-0.3962,0.6846944)--	(-0.3832,0.6667622)--	(-0.3702,0.649168)--	(-0.3572,0.6319118)--	(-0.3442,0.6149936)--	(-0.3312,0.5984134)--	(-0.3182,0.5821712)--	(-0.3052,0.566267)--	(-0.2922,0.5507008)--	(-0.2792,0.5354726)--	(-0.2662,0.5205824)--	(-0.2532,0.5060302)--	(-0.2402,0.491816)--	(-0.2272,0.4779398)--	(-0.2142,0.4644016)--	(-0.2012,0.4512014)--	(-0.1882,0.4383392)--	(-0.1752,0.425815)--	(-0.1622,0.4136288)--	(-0.1492,0.4018156)--	(-0.1362,0.3910601)--	(-0.1232,0.3809944)--	(-0.1102,0.3716184)--	(-0.0972,0.3629323)--	(-0.0842,0.354936)--	(-0.0712,0.350784)--	(-0.0582,0.3500202)--	(-0.0452,0.351369)--	(-0.0322,0.3548303)--	(-0.0192,0.360404)--	(-0.0062,0.3680903)--	(0.0068,0.377889)--	(0.0198,0.3898003)--	(0.0328,0.403824)--	(0.0458,0.462264)--	(0.0588,0.4888794)--	(0.0718,0.4688288)--	(0.0848,0.4515053)--	(0.0978,0.455264)--	(0.1108,0.4594994)--	(0.1238,0.4642114)--	(0.1368,0.4694001)--	(0.1498,0.4750654)--	(0.1628,0.4574722)--	(0.1758,0.4395447)--	(0.1888,0.4216501)--	(0.2018,0.4037886)--	(0.2148,0.3859601)--	(0.2278,0.3681645)--	(0.2408,0.350402)--	(0.2538,0.3326725)--	(0.2668,0.314976)--	(0.2798,0.2989665)--	(0.2928,0.2849025)--	(0.3058,0.284106)--	(0.3188,0.2835928)--	(0.3318,0.2833049)--	(0.3448,0.2832421)--	(0.3578,0.2834047)--	(0.3708,0.2837924)--	(0.3838,0.2844054)--	(0.3968,0.2884229)--	(0.4098,0.3002715)--	(0.4228,0.3128099)--	(0.4358,0.326038)--	(0.4488,0.339956)--	(0.4618,0.3394426)--	(0.4748,0.3375286)--	(0.4878,0.3357249)--	(0.5008,0.3340316)--	(0.5138,0.3324487)--	(0.5268,0.3309762)--	(0.5398,0.329614)--	(0.5528,0.3283622)--	(0.5658,0.3272207)--	(0.5788,0.3261896)--	(0.5918,0.3252689)--	(0.6048,0.3244585)--	(0.6178,0.3237586)--	(0.6308,0.3231689)--	(0.6438,0.3226897)--	(0.6568,0.3273062)--	(0.6698,0.336752)--	(0.6828,0.3465358)--	(0.6958,0.3566576)--	(0.7088,0.3671174)--	(0.7218,0.3779152)--	(0.7348,0.389051)--	(0.7478,0.4005248)--	(0.7608,0.4123366)--	(0.7738,0.4244864)--	(0.7868,0.4369742)--	(0.7998,0.4498);

    \draw[smooth,blue,dotted, thick] (-0.5,0.2857653)--	(-0.4872,0.2857653)--	(-0.4742,0.2857653)--	(-0.4612,0.2857653)--	(-0.4482,0.2857653)--	(-0.4352,0.2857653)--	(-0.4222,0.2857653)--	(-0.4092,0.2857653)--	(-0.3962,0.2857653)--	(-0.3832,0.2857653)--	(-0.3702,0.2857653)--	(-0.3572,0.2857653)--	(-0.3442,0.2857653)--	(-0.3312,0.2857653)--	(-0.3182,0.2857653)--	(-0.3052,0.2857653)--	(-0.2922,0.2857653)--	(-0.2792,0.2857653)--	(-0.2662,0.2857653)--	(-0.2532,0.2857653)--	(-0.2402,0.2857653)--	(-0.2272,0.2857653)--	(-0.2142,0.2857653)--	(-0.2012,0.2857653)--	(-0.1882,0.2857653)--	(-0.1752,0.2857653)--	(-0.1622,0.2857653)--	(-0.1492,0.2860524)--	(-0.1362,0.2907854)--	(-0.1232,0.295645)--	(-0.1102,0.3006314)--	(-0.0972,0.3057444)--	(-0.0842,0.3109842)--	(-0.0712,0.3191236)--	(-0.0582,0.3289323)--	(-0.0452,0.339129)--	(-0.0322,0.3497137)--	(-0.0192,0.3606864)--	(-0.0062,0.3720472)--	(0.0068,0.3837959)--	(0.0198,0.3959327)--	(0.0328,0.4129823)--	(0.0458,0.4688423)--	(0.0588,0.474157)--	(0.0718,0.4348372)--	(0.0848,0.4015112)--	(0.0978,0.3928527)--	(0.1108,0.3843844)--	(0.1238,0.3761061)--	(0.1368,0.368018)--	(0.1498,0.36012)--	(0.1628,0.3524122)--	(0.1758,0.3448944)--	(0.1888,0.3375668)--	(0.2018,0.3304293)--	(0.2148,0.323482)--	(0.2278,0.3167247)--	(0.2408,0.3101576)--	(0.2538,0.3037806)--	(0.2668,0.2975938)--	(0.2798,0.291597)--	(0.2928,0.2857904)--	(0.3058,0.2857653)--	(0.3188,0.2857653)--	(0.3318,0.2857653)--	(0.3448,0.2857653)--	(0.3578,0.2857653)--	(0.3708,0.2857653)--	(0.3838,0.2857653)--	(0.3968,0.2857653)--	(0.4098,0.2857653)--	(0.4228,0.2857653)--	(0.4358,0.2857653)--	(0.4488,0.2857653)--	(0.4618,0.2857653)--	(0.4748,0.2857653)--	(0.4878,0.2857653)--	(0.5008,0.2857653)--	(0.5138,0.2857653)--	(0.5268,0.2857653)--	(0.5398,0.2857653)--	(0.5528,0.2857653)--	(0.5658,0.2857653)--	(0.5788,0.2857653)--	(0.5918,0.2857653)--	(0.6048,0.2857653)--	(0.6178,0.2857653)--	(0.6308,0.2857653)--	(0.6438,0.2857653)--	(0.6568,0.2857653)--	(0.6698,0.2857653)--	(0.6828,0.2857653)--	(0.6958,0.2857653)--	(0.7088,0.2857653)--	(0.7218,0.2857653)--	(0.7348,0.2857653)--	(0.7478,0.2857653)--	(0.7608,0.2857653)--	(0.7738,0.2857653)--	(0.7868,0.2857653)--	(0.7998,0.2857653);

\draw[smooth,blue,dashdotted, thick] (-0.5,0.2543367)--	(-0.4872,0.2543367)--	(-0.4742,0.2543367)--	(-0.4612,0.2543367)--	(-0.4482,0.2543367)--	(-0.4352,0.2543367)--	(-0.4222,0.2543367)--	(-0.4092,0.2543367)--	(-0.3962,0.2543367)--	(-0.3832,0.2543367)--	(-0.3702,0.2543367)--	(-0.3572,0.2543367)--	(-0.3442,0.2543367)--	(-0.3312,0.2543367)--	(-0.3182,0.2543367)--	(-0.3052,0.2543367)--	(-0.2922,0.2543367)--	(-0.2792,0.2543367)--	(-0.2662,0.2543367)--	(-0.2532,0.2543367)--	(-0.2402,0.2543367)--	(-0.2272,0.2543367)--	(-0.2142,0.2543367)--	(-0.2012,0.2543367)--	(-0.1882,0.2543367)--	(-0.1752,0.2543367)--	(-0.1622,0.2543367)--	(-0.1492,0.2544522)--	(-0.1362,0.2563347)--	(-0.1232,0.2582286)--	(-0.1102,0.2601339)--	(-0.0972,0.2620506)--	(-0.0842,0.2639787)--	(-0.0712,0.2723436)--	(-0.0582,0.2841023)--	(-0.0452,0.296249)--	(-0.0322,0.3115061)--	(-0.0192,0.3302041)--	(-0.0062,0.349694)--	(0.0068,0.3699757)--	(0.0198,0.3910492)--	(0.0328,0.4152822)--	(0.0458,0.4825423)--	(0.0588,0.5055142)--	(0.0718,0.4810593)--	(0.0848,0.4603301)--	(0.0978,0.4632624)--	(0.1108,0.4662065)--	(0.1238,0.4691626)--	(0.1368,0.4721305)--	(0.1498,0.4751103)--	(0.1628,0.4548585)--	(0.1758,0.4348214)--	(0.1888,0.4153665)--	(0.2018,0.3964939)--	(0.2148,0.3782035)--	(0.2278,0.3604954)--	(0.2408,0.3433696)--	(0.2538,0.326826)--	(0.2668,0.3108646)--	(0.2798,0.2984766)--	(0.2928,0.2898843)--	(0.3058,0.2850439)--	(0.3188,0.2802817)--	(0.3318,0.2755815)--	(0.3448,0.2709434)--	(0.3578,0.2663674)--	(0.3708,0.2618535)--	(0.3838,0.2574017)--	(0.3968,0.2543367)--	(0.4098,0.2543367)--	(0.4228,0.2543367)--	(0.4358,0.2543367)--	(0.4488,0.2543367)--	(0.4618,0.2543367)--	(0.4748,0.2543367)--	(0.4878,0.2543367)--	(0.5008,0.2543367)--	(0.5138,0.2543367)--	(0.5268,0.2543367)--	(0.5398,0.2543367)--	(0.5528,0.2543367)--	(0.5658,0.2543367)--	(0.5788,0.2543367)--	(0.5918,0.2543367)--	(0.6048,0.2543367)--	(0.6178,0.2543367)--	(0.6308,0.2543367)--	(0.6438,0.2543367)--	(0.6568,0.2543367)--	(0.6698,0.2543367)--	(0.6828,0.2543367)--	(0.6958,0.2543367)--	(0.7088,0.2543367)--	(0.7218,0.2543367)--	(0.7348,0.2543367)--	(0.7478,0.2543367)--	(0.7608,0.2543367)--	(0.7738,0.2543367)--	(0.7868,0.2543367)--	(0.7998,0.2543367);

\draw[smooth,blue,dashed, thick](-0.5,0.2225)--	(-0.4872,0.2225)--	(-0.4742,0.2225)--	(-0.4612,0.2225)--	(-0.4482,0.2225)--	(-0.4352,0.2225)--	(-0.4222,0.2225)--	(-0.4092,0.2225)--	(-0.3962,0.2225)--	(-0.3832,0.2225)--	(-0.3702,0.2225)--	(-0.3572,0.2225)--	(-0.3442,0.2225)--	(-0.3312,0.2225)--	(-0.3182,0.2225)--	(-0.3052,0.2225)--	(-0.2922,0.2225)--	(-0.2792,0.2225)--	(-0.2662,0.2225)--	(-0.2532,0.2225)--	(-0.2402,0.2225)--	(-0.2272,0.2225)--	(-0.2142,0.2225)--	(-0.2012,0.2225)--	(-0.1882,0.2225)--	(-0.1752,0.2225)--	(-0.1622,0.2225)--	(-0.1492,0.2225)--	(-0.1362,0.2225)--	(-0.1232,0.2225)--	(-0.1102,0.2225)--	(-0.0972,0.2225)--	(-0.0842,0.2225)--	(-0.0712,0.2225)--	(-0.0582,0.2225)--	(-0.0452,0.2225)--	(-0.0322,0.226369)--	(-0.0192,0.2342003)--	(-0.0062,0.2421029)--	(0.0068,0.2500767)--	(0.0198,0.2581218)--	(0.0328,0.2718316)--	(0.0458,0.3453823)--	(0.0588,0.3727462)--	(0.0718,0.3436113)--	(0.0848,0.3192716)--	(0.0978,0.3270288)--	(0.1108,0.3348436)--	(0.1238,0.3427158)--	(0.1368,0.3506456)--	(0.1498,0.3586329)--	(0.1628,0.3494664)--	(0.1758,0.3401111)--	(0.1888,0.3308361)--	(0.2018,0.3216415)--	(0.2148,0.3125271)--	(0.2278,0.3034931)--	(0.2408,0.2945395)--	(0.2538,0.2856661)--	(0.2668,0.2768731)--	(0.2798,0.2719889)--	(0.2928,0.2716062)--	(0.3058,0.275968)--	(0.3188,0.280371)--	(0.3318,0.2847943)--	(0.3448,0.2892378)--	(0.3578,0.2937017)--	(0.3708,0.2981857)--	(0.3838,0.3026901)--	(0.3968,0.3082468)--	(0.4098,0.3162217)--	(0.4228,0.3242587)--	(0.4358,0.3323579)--	(0.4488,0.340519)--	(0.4618,0.3338599)--	(0.4748,0.3257493)--	(0.4878,0.3177008)--	(0.5008,0.3097144)--	(0.5138,0.3017901)--	(0.5268,0.2939278)--	(0.5398,0.2861277)--	(0.5528,0.2783896)--	(0.5658,0.2707136)--	(0.5788,0.2630997)--	(0.5918,0.2555479)--	(0.6048,0.2480581)--	(0.6178,0.2406304)--	(0.6308,0.2332649)--	(0.6438,0.2259613)--	(0.6568,0.2225)--	(0.6698,0.2225)--	(0.6828,0.2225)--	(0.6958,0.2225)--	(0.7088,0.2225)--	(0.7218,0.2225)--	(0.7348,0.2225)--	(0.7478,0.2225)--	(0.7608,0.2225)--	(0.7738,0.2225)--	(0.7868,0.2225)--	(0.7998,0.2225);

% 	 \draw[smooth,dotted,yellow,thick] (0.05,\ymin) -- (0.05,\ymax);
%     \draw[smooth,dotted,yellow,thick] (0.15,\ymin) -- (0.15,\ymax);
%     \draw[smooth,dotted,yellow,thick] (0.45,\ymin) -- (0.45,\ymax);
%     
%     \draw[smooth,dotted,black,thick] (0.293,\ymin) -- (0.293,\ymax);
%     \draw[smooth,dotted,black,thick] (0.393,\ymin) -- (0.393,\ymax);
%     \draw[smooth,dotted,black,thick] (0.65,\ymin) -- (0.65,\ymax);
%     
%     \draw[smooth,dotted,red,thick] (-0.15,\ymin) -- (-0.15,\ymax);
%     \draw[smooth,dotted,red,thick] (-0.08,\ymin) -- (-0.08,\ymax);
%     \draw[smooth,dotted,red,thick] (-0.0387,\ymin) -- (-0.0387,\ymax);
%     \draw[smooth,dotted,red,thick] (0.04,\ymin) -- (0.04,\ymax);
%     \draw[smooth,dotted,red,thick] (0.2738,\ymin) -- (0.2738,\ymax);  
     
    \begin{customlegend}
    [legend entries={ $c_1(\mathbf{a})$,$c_2(\mathbf{a})$,$c_3(\mathbf{a})$,$c_4(\mathbf{a})$},
    legend style={at={(0.65,1.05)}}]   
    \addlegendimage{red,thick}     
    \addlegendimage{blue,dotted,thick}    
    \addlegendimage{blue,dashdotted,thick}  
    \addlegendimage{blue,dashed,thick}  
    \end{customlegend}
\end{tikzpicture}
\caption{An illustration of the effect on the cost of changing the arrival time of a single user while keeping the arrival times of the other three users constant. The horizontal axis is the arrival time of user $1$. The vertical axis shows the cost of all users. The parameter values of the example are: $\alpha=1.5$, $\beta=5$, and $d_i^*=0.5$ for $i=1,\ldots,4$.}
\label{fig:cost_arrival}
\end{figure}

In \cite{RN2015} it was shown that in the socially optimal arrival schedule the order of the users is in accordance with their departure times. For small values of $N$ ($\leq 15$) the exact socially optimal arrival scheduled can be computed using an exhaustive algorithm. For the general case, a heuristic search algorithm was suggested and shown to be efficient for larger values of $N$ ($\leq 500$). The key underlying property which makes scheduling difficult using these travel dynamics is the fact that the departure and cost functions behave differently given the combined order permutation of the arrival and departure times. From a combinatorial perspective, the number of possible permutations is exponential with the number of users. Unsurprisingly, this property will make the game of arrivals, defined in the next section, computationally difficult to solve as well. 

%%%%%%%%%%%%%%%%%%%%%%%%%%%%%
\section{Ordered arrival game}\label{sec:seq_game}
The game of arrivals is defined as follows: every user $i=1,\ldots,N$ chooses an arrival time $a_i\in[a_{i-1},\infty)$, where $a_0:=-\infty$. Let $\mathbf{a}_{-i}$ denote the arrival times vector of all users excluding $i$. The set of best response actions of user $i\in\mathcal{N}$ to the arrival times of all other users is
\[
\mathcal{BR}_i(\mathbf{a}_{-i}):=\argmin_{a\in[a_{i-1},\infty)}\{c_i(a,\mathbf{a}_{-i})\}.
\]
Note that $\mathcal{BR}$ is a set that may contain multiple minimizers, as the cost function is not convex.

We consider two possible game formulations. The first is a simultaneous game (Cournot) in which all users decide at the same time, without inspecting the arrival time of previous users. In this case we impose the order of arrivals on the game in the following manner: if user $i$ selects time $a_i<a_{j}$ for some $i>j$ then his effective arrival time will be $a_i=a_{j}$. We limit our analysis to pure strategies throughout this work, although mixed strategy equilibria are possible and are likely to exist as well.

\begin{definition}\label{def:CNE}
A pure strategy Cournot Nash equilibrium (CNE) is a vector of ordered arrivals, $\mathbf{a}\in\mathbbm{R}^N$, such that $a_i\in\mathcal{BR}(\mathbf{a}_{-i}),\ \forall i\in\mathcal{N}$.
\end{definition}

An alternative, and perhaps more natural, formulation for this model is a sequential game (Stackelberg), in which users make their decisions after observing the arrival times of the previous users. In this setting the action space of user $i\in\mathcal{N}$ is a function, $b_i(a_0,\ldots,a_{i-1}):\mathbbm{R}^{i-1}:\rightarrow\mathbbm{R}$, specifying his arrival time given the arrival time of all preceding users. The definition of the best response action is therefore not sufficient to define the equilibrium. We provide a definition using a constructive backward induction method.
 
\begin{definition}\label{def:SPNE}
A pure strategy Subgame Perfect Nash equilibrium (SPNE) is a collection of functions, $\{b_i:\mathbbm{R}^{i-1}:\rightarrow\mathbbm{R}\}$, that satisfies the backward induction:
\begin{itemize}
\item The initial condition is the function $b_N:\mathbbm{R}^{N-1}\rightarrow \mathbbm{R}$, defined by choosing an arbitrary element $b_N\in\mathcal{BR}(a_1,\ldots,a_{N-1})$ for any $\mathbf{a}_{-N}$.
\item For $i<N$,
\[
b_{i}(a_0,\ldots,a_{i-1})\in\argmin_{a\in[a_{i-1},\infty)}\{c_i(a_1,\ldots,a_{i-1},a,B_i(a))\},
\]
where
\[
B_i(a)=\left(b_{i+1}(a_1,\ldots,a_{i-1},a),\ldots,b_N\left(a_1,\ldots,a,\ldots,b_{N-1}(a_1,\ldots,a,\ldots)\right)\right),
\] 
is the collection of best response functions of users $j>i$. If the minimizer at any step is not unique then a minimizing element can be arbitrarily selected.
\end{itemize}
\end{definition}

Unfortunately, this method is only tractable for small values of $N$, as we will show in the next sections. The CNE is a more coarse equilibrium solution, as it is not subgame perfect. Nevertheless, it is a Nash equilibrium in both games that represents a stable schedule of arrivals in a sense that no single user has incentive to deviate. Take care however, and note that a CNE is not necessarily an equilibrium path of any SPNE. This is due to the fact that a CNE may include non-credible threats; the strategy instructs a user to play an action that is sub-optimal for some actions of the preceding users, but because all play a fixed action these situations are not realised. In the next section we shall show this can happen for a non-degenerate range of parameters in a two-user example.

Computing a CNE is a more modest goal, but still computationally challenging for games with a large number of users. This will lead us to suggest a heuristic method to find such equilibrium paths in Section \ref{sec:OBRA}. Numerical analysis using the above method will be presented in Section \ref{sec:numerical}, along with a comparison to the socially optimal solution of \cite{RN2015}. Before proceeding to the general analysis for any number of users we will commence by explicitly solving the two-user game, which provides some important insight about the model dynamics and difficulties. 

%%%%%%%%%%%%%%%%%%%%%%%%%%%%%
\section{Two-user game}\label{sec:example_N2}
Suppose that there are only two users arriving to the system. We assume that they both have the same desired departure time and select it to be zero without loss of generality, $d_1^*=d_2^*=0$. We arbitrarily set user $1$ to be the leader. We first analyse the SPNE of the sequential game, with the results summarized in Proposition \ref{prop:N2_SPNE}. The proof explicitly constructs all possible equilibria using a backward induction argument. This will be followed by the computing the CNE of the simultaneous game, which is fully characterised in Proposition \ref{prop:N2_CNE}. 

From equation \eqref{eq:d_recursive} we have that for $a_1\leq a_2$ the departure times are given by:
\[
d_1(a_1,a_2)=
\left\{
	\begin{array}{ll}
		a_1+\frac{1}{\beta} & \mbox{, } a_1\leq a_2-\frac{1}{\beta}, \\
		\frac{1+\beta a_1-\alpha a_2}{\beta-\alpha} & \mbox{, } a_2-\frac{1}{\beta}\leq a_1\leq a_2,
	\end{array}
\right.
\]
and
\[
d_2(a_1,a_2)=
\left\{
	\begin{array}{ll}
		a_2+\frac{1}{\beta} & \mbox{, } a_1\leq a_2-\frac{1}{\beta}, \\
		\frac{1+\beta a_2+\alpha(d_1(a_1,a_2)-a_2)}{\beta} & \mbox{, } a_2-\frac{1}{\beta}\leq a_1\leq a_2.
	\end{array}
\right.
\]

Recall that the free flow speed is $\beta$ and the travel time when a user is alone in the system is $\frac{1}{\beta}$. Hence, the condition $a_2>a_1+\frac{1}{\beta}$ can be interpreted as: the second arrival takes place after the first departure.

A unique feature of the two-user model is that both users spend the same amount of time in the system. This is because during the time that both users are in the system they travel at the same speed, $\beta-\alpha$, hence they travel the same distance during this time. For the remaining distance, which is identical for both, they travel at the free flow speed and again require the same time to complete the journey. Therefore,
\[
d_2-a_2=d_1-a_1=\left\{
	\begin{array}{ll}
		\frac{1}{\beta} & \mbox{, } a_1\leq a_2-\frac{1}{\beta}, \\
		\frac{1-\alpha(a_2-a_1)}{\beta-\alpha} & \mbox{, } a_2-\frac{1}{\beta}\leq a_1\leq a_2.
	\end{array}
\right.
\]

The cost of user $2$ given that user $1$ is arriving at $a_1$ is then given by \eqref{eq:costUserShort}:
\[
c_2(a_1,a_2)=
\left\{
	\begin{array}{ll}
		\left(\frac{1+\beta a_2}{\beta}\right)^2+\frac{\gamma}{\beta} & \mbox{, } a_2\geq a_1+\frac{1}{\beta}, \\
		\left(\frac{1+\beta a_2+\alpha(d_1(a_1,a_2)-a_2)}{\beta}\right)^2+\gamma\frac{1-\alpha(a_2-a_1)}{\beta-\alpha} & \mbox{, } a_1\leq a_2\leq a_1+\frac{1}{\beta}.
	\end{array}
\right.
\]

If $a_1\leq -\frac{2}{\beta}$ then user $2$ can obtain a minimal cost of $\frac{\gamma}{\beta}$ by arriving at exactly $-\frac{1}{\beta}$. For this best response to be unique we need to assume $\gamma>0$. Otherwise, if $\gamma=0$ and $-\frac{1}{\beta-\alpha}\leq a_1<-\frac{2}{\beta}$ then user $2$ can obtain a cost of zero by arriving at either $-\frac{1}{\beta}$ or at $-\frac{1+\alpha a_1}{\beta-2\alpha}$ which both yield $d_2=0$. For instance, $a_1=a_2=-\frac{1}{\beta-\alpha}$ guarantees that both travel at the slowest speed and depart together at zero.

If $a_1> -\frac{2}{\beta}$, the minimal cost satisfying $a_2\geq a_1+\frac{1}{\beta}$ is obtained by $a_2=a_1+\frac{1}{\beta}$, i.e., arriving immediately upon user $1$'s departure, for arriving later will increase the tardiness cost and not change the travel time. For the second region, $a_1\leq a_2\leq a_1+\frac{1}{\beta}$, the optimal arrival of user $2$ is given by solving the corresponding one dimensional constrained quadratic program. After plugging in $d_1(a_1,a_2)$ and applying simple algebra we have that the quadratic program that needs to be solved is
\[
\min_{a_1\leq x\leq a_1+\frac{1}{\beta}}\left\lbrace x^2U+xV+W\right\rbrace,
\]
where 
\[
\begin{split}
U &= \left(\frac{\beta-2\alpha}{\beta-\alpha}\right)^2 , \\
V &=  \frac{2(\beta-2\alpha)(1+\alpha a_1)-\gamma\alpha(\beta-\alpha)}{(\beta-\alpha)^2},
\end{split}
\]
and $W$ is a constant with no impact on the optimization. We denote the unconstrained minimum of the quadratic program by $X=-\frac{V}{2U}$. The best response of user $2$ to $a_1$ can then be explicitly stated as
\begin{equation}\label{eq:br_user2}
b_2(a_1)=\left\{
	\begin{array}{ll}
		-\frac{1}{\beta} & \mbox{, } a_1\leq -\frac{2}{\beta}, \\
		 x^*(a_1) & \mbox{, } a_1>-\frac{2}{\beta},
	\end{array}
\right.
\end{equation}
where
\[
x^*(a):=\left\{
	\begin{array}{ll}
		a & \mbox{, } a\geq X, \\
		X & \mbox{, } X-\frac{1}{\beta}<a<X, \\
		 a+\frac{1}{\beta} & \mbox{, } a\leq X-\frac{1}{\beta}.
	\end{array}
\right.
\]
Note that if $\alpha=\frac{\beta}{2}$ then $U=0$ (the cost function is linear) and we define $x^*(a):=a+\frac{1}{\beta}$. Furthermore, as explained above, if $\gamma=0$ then the best response is not unique for $a_1\in\left[-\frac{1}{\beta-\alpha},-\frac{2}{\beta}\right]$. Specifically, arriving at $-\frac{1+\alpha a_1}{\beta-2\alpha}$ yields a cost of zero, and so does arriving at $-\frac{1}{\beta}$ as prescribed by \eqref{eq:br_user2}.

The optimal choice of user $1$ if user $2$ is playing $b:\mathbbm{R}\rightarrow\mathbbm{R}$ is the collection of arrival times,
\[
b_1(a_0)=\argmin_{a\geq a_0}\{c_1(a,b_2(a))\},
\]
where $c_1$ is the cost function \eqref{eq:costUserShort} and $a_0=-\infty$ as in Definition \ref{def:SPNE}. We are now ready to fully characterise the SPNE.

\begin{proposition}\label{prop:N2_SPNE}
For $N=2$ and $\mathbf{d}^*=(0,0)$ the SPNE are characterised as follows.
\begin{enumerate}
\item In any case, the strategy of user $2$ is to arrive according to $b_2(a_1)$ for any arrival time of user $1$, $a_1$.
\item If $\gamma=0$ then 
\begin{enumerate}
\item if $\alpha<\frac{\beta}{2}$ the unique SPNE is $a_1=-\frac{1}{\beta-\alpha}$.
\item if $\alpha\geq\frac{\beta}{2}$ then there are two SPNE: $a_1=-\frac{1}{\beta-\alpha}$ and $a_1=-\frac{1}{\beta}$. The resulting equilibrium paths (realisations) are $\left(-\frac{1}{\beta-\alpha},-\frac{1}{\beta-\alpha}\right)$ and $\left(-\frac{1}{\beta},0\right)$, with respective equilibrium costs $(0,0)$ and $\left(0,\frac{1}{\beta^2}\right)$.
\end{enumerate}
\item If $\gamma>0$ then there is a unique SPNE such that
\begin{enumerate}
\item if $\alpha<\frac{\beta}{2}$ then $a_1=\argmin_{a>-\infty}\{c_1(a,b_2(a))\}\in(-\frac{2}{\beta},-\frac{1}{\beta}]$.
\item if $\alpha\geq \frac{\beta}{2}$ then $a_1=-\frac{1}{\beta}$. The resulting equilibrium path is $\left(-\frac{1}{\beta},0\right)$ with costs $\left(\frac{\gamma}{\beta},\frac{1}{\beta^2}+\frac{\gamma}{\beta}\right)$.
\end{enumerate}
\end{enumerate}
\end{proposition}
\begin{proof}
We have already established that $b_2(a_1)$ is the best response strategy of user $2$, hence this is the strategy played in any equilibrium. Computing the optimal response of user $1$ to this strategy requires solving two quadratic programs: the first is a trivial one for $a_1\leq -\frac{2}{\beta}$, and the second one gives $x^*$. 

If $\alpha\geq\frac{\beta}{2}$ then user $1$ can make sure user $2$ arrives at $0$, resulting in the minimal cost of $\frac{\gamma}{\beta}$ for user $1$. This is achieved by arriving at $a_1=-\frac{1}{\beta}$, which satisfies $a_1\leq X-\frac{1}{\beta}$. If $\gamma=0$ then arriving at $-\frac{1}{\beta-\alpha}$ also achieves the minimal cost for user 1, which is zero in this case. Therefore, if $\gamma=0$ both $a_1=-\frac{1}{\beta}$ and $a_1=-\frac{1}{\beta-\alpha}$ are SPNE (together with $b_2(a)$).  

If $\alpha<\frac{\beta}{2}$ then there is a unique minimizer, $-\frac{2}{\beta}<a\leq-\frac{1}{\beta}$, to $c_1(a,b_2(a))$. This is seen by considering the unconstrained quadratic program solution for user 2,
\[
X=-\frac{2(\beta-2\alpha)(1+\alpha a)-\gamma\alpha(\beta-\alpha)}{2(\beta-2\alpha)^2}.
\]
If $\gamma=0$ then $X<0$ for $a=-\frac{1}{\beta}$, implying that $b_2\left(-\frac{1}{\beta}\right)<0$. Thus, we have that $a_1=-\frac{1}{\beta-\alpha}$ is the unique SPNE in this case. If $\gamma>0$ then $X$ is a monotone decreasing function of $a$, hence $b_2(a)$ is monotone decreasing as well and $c_1(a,b_2(a))$ has a unique minimizer.  \qed
\end{proof}

In Figure \ref{fig:arrivals_N2} we illustrate the realised arrival times of the two users according to an SPNE for varying values of the travel time penalty, $\gamma$, and the linear slowdown coefficient, $\alpha$.

\begin{figure}[H]
\begin{subfigure}{.49\linewidth}
\begin{tikzpicture}[xscale=0.7,yscale=1.1]
  \def\xmin{0}
  \def\xmax{8.1}
  \def\ymin{-2}
  \def\ymax{0.1}
    \draw[->] (\xmin,\ymin) -- (\xmax,\ymin) node[right] {$\gamma$} ;
    \draw[->] (0,\ymin) -- (0,\ymax) node[above] {$a_i$} ;
    \foreach \x in {0,2,4,6,8}
    \node at (\x,\ymin) [below] {\x};
    \foreach \y in {-2,-1,0}
    \node at (0,\y) [left] {\y};

	\draw[smooth,blue,thick] (0,-1.1112)--	(0.5,-1.1322)--	(1,-1.1532)--	(1.5,-1.1752)--	(2,-1.1962)--	(2.5,-1.2172)--	(3,-1.2382)--	(3.5,-1.2592)--	(4,-1.2812)--	(4.5,-1.3022)--	(5,-1.3232)--	(5.5,-1.3442)--	(6,-1.3662)--	(6.5,-1.3872)--	(7,-1.4082)--	(7.5,-1.4292)--	(8,-1.4502);

	\draw[smooth,red,dashdotted, thick] (0,-1.1111)--	(0.5,-1.0733)--	(1,-1.0355)--	(1.5,-0.9976)--	(2,-0.9598)--	(2.5,-0.9221)--	(3,-0.8843)--	(3.5,-0.8465)--	(4,-0.8086)--	(4.5,-0.7708)--	(5,-0.733)--	(5.5,-0.6953)--	(6,-0.6573)--	(6.5,-0.6196)--	(7,-0.5818)--	(7.5,-0.544)--	(8,-0.5062);
\end{tikzpicture}
\caption*{$\alpha=0.1$}
\end{subfigure}
\begin{subfigure}{.49\linewidth}
\begin{tikzpicture}[xscale=0.7,yscale=1.1]
  \def\xmin{0}
  \def\xmax{8.1}
  \def\ymin{-2}
  \def\ymax{0.1}
    \draw[->] (\xmin,\ymin) -- (\xmax,\ymin) node[right] {$\gamma$} ;
    \draw[->] (0,\ymin) -- (0,\ymax) node[above] {$a_i$} ;
    \foreach \x in {0,2,4,6,8}
    \node at (\x,\ymin) [below] {\x};
    \foreach \y in {-2,-1,0}
    \node at (0,\y) [left] {\y};

	\draw[smooth,blue,thick] (0,-1.25)--	(0.5,-1.276)--	(1,-1.302)--	(1.5,-1.328)--	(2,-1.354)--	(2.5,-1.38)--	(3,-1.406)--	(3.5,-1.416)--	(4,-1.333)--	(4.5,-1.25)--	(5,-1.167)--	(5.5,-1.084)--	(6,-1)--	(6.5,-1)--	(7,-1)--	(7.5,-1)--	(8,-1);
	
	\draw[smooth,red,dashdotted, thick] (0,-1.25)--	(0.5,-1.1302)--	(1,-1.0104)--	(1.5,-0.8907)--	(2,-0.7709)--	(2.5,-0.6511)--	(3,-0.5313)--	(3.5,-0.4169)--	(4,-0.3334)--	(4.5,-0.25)--	(5,-0.167)--	(5.5,-0.084)--	(6,0)--	(6.5,0)--	(7,0)--	(7.5,0)--	(8,0);
\end{tikzpicture}
\caption*{$\alpha=0.2$}
\end{subfigure}

\begin{subfigure}{.49\linewidth}
\begin{tikzpicture}[xscale=0.7,yscale=1.1]
  \def\xmin{0}
  \def\xmax{8.1}
  \def\ymin{-2}
  \def\ymax{0.1}
    \draw[->] (\xmin,\ymin) -- (\xmax,\ymin) node[right] {$\gamma$} ;
    \draw[->] (0,\ymin) -- (0,\ymax) node[above] {$a_i$} ;
    \foreach \x in {0,2,4,6,8}
    \node at (\x,\ymin) [below] {\x};
    \foreach \y in {-2,-1,0}
    \node at (0,\y) [left] {\y};

	\draw[smooth,blue,thick] (0,-1.6663)--	(0.5,-1.3893)--	(1,-1.0003)--	(1.5,-1.0003)--	(2,-1.0003)--	(2.5,-1.0003)--	(3,-1.0003)--	(3.5,-1.0003)--	(4,-1.0003)--	(4.5,-1.0003)--	(5,-1.0003)--	(5.5,-1.0003)--	(6,-1.0003)--	(6.5,-1.0003)--	(7,-1.0003)--	(7.5,-1.0003)--	(8,-1.0003);
	
	\draw[smooth,red,dashdotted, thick] (0,-1.6663)--	(0.5,-0.7213)--	(1,-0.0003)--	(1.5,-0.0003)--	(2,-0.0003)--	(2.5,-0.0003)--	(3,-0.0003)--	(3.5,-0.0003)--	(4,-0.0003)--	(4.5,-0.0003)--	(5,-0.0003)--	(5.5,-0.0003)--	(6,-0.0003)--	(6.5,-0.0003)--	(7,-0.0003)--	(7.5,-0.0003)--	(8,-0.0003);

\end{tikzpicture}
\caption*{$\alpha=0.4$}
\end{subfigure}
\begin{subfigure}{.49\linewidth}
\begin{tikzpicture}[xscale=0.7,yscale=1.1]
  \def\xmin{0}
  \def\xmax{8.1}
  \def\ymin{-2}
  \def\ymax{0.1}
    \draw[->] (\xmin,\ymin) -- (\xmax,\ymin) node[right] {$\gamma$} ;
    \draw[->] (0,\ymin) -- (0,\ymax) node[above] {$a_i$} ;
    \foreach \x in {0,2,4,6,8}
    \node at (\x,\ymin) [below] {\x};
    \foreach \y in {-2,-1,0}
    \node at (0,\y) [left] {\y};

	\draw[smooth,blue,thick] (0,-1)--	(8,-1);
	
	\draw[smooth,red,dashdotted,thick] (0,0)--	(8,0);
\end{tikzpicture}
\caption*{$\alpha=0.6$}
\end{subfigure}

\centering
\begin{subfigure}[b]{.2\linewidth}
\begin{tikzpicture}[scale=10]
    \begin{customlegend}
    [legend entries={User $1$,User $2$},
    legend style={at={(1.4,0.5)}}]
    \addlegendimage{blue,thick}
    \addlegendimage{red,thick,dashdotted}
    \end{customlegend}
\end{tikzpicture}
\end{subfigure}
\caption{SPNE arrival times density for $\gamma\in[0,8]$ with $\beta=1$, $\mathbf{d}^*=(0,0)$, and different values of $\alpha$.}
\label{fig:arrivals_N2}
\end{figure}
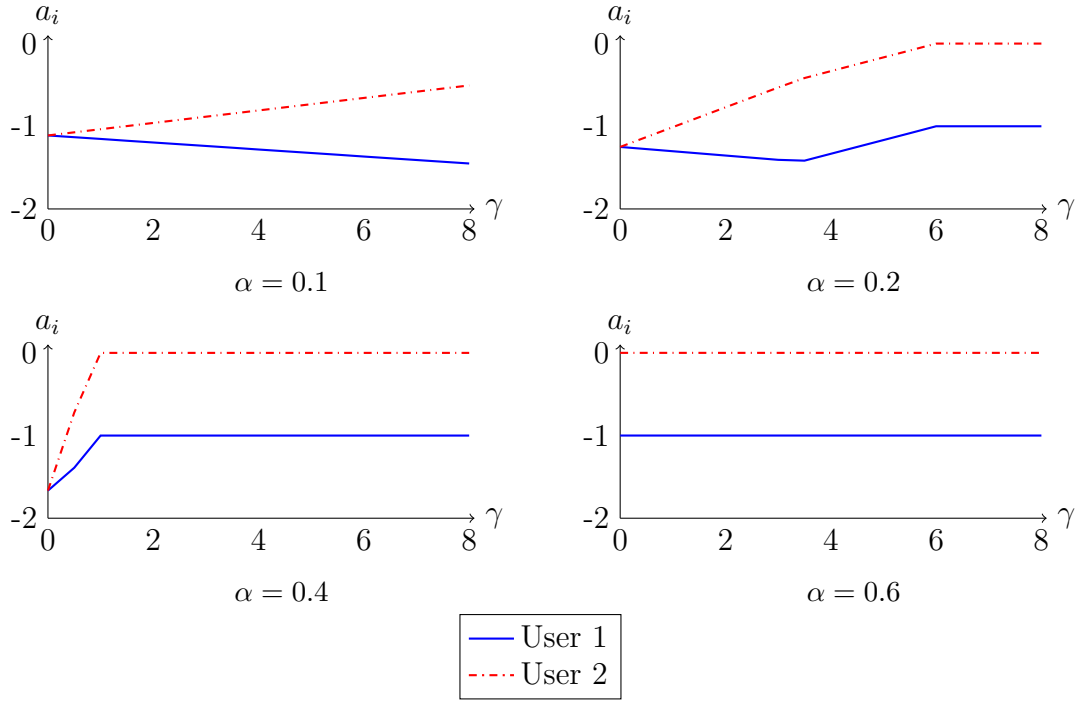

In order to find the Cournot Nash Equilibrium points, the best response of user $1$ for any arrival time of user $2$ needs to be computed. Note that $b_2(a_1)$ as defined above remains the best response for user $2$ given any action $a_1$ of user $1$. As before, by considering the model dynamics we get that the best response is given by solving the respective quadratic program for each region of the cost function \eqref{eq:costUserShort},
\[
c_1(a_1,a_2)=
\left\{
	\begin{array}{ll}
		\left(\frac{1+\beta a_1}{\beta}\right)^2+\frac{\gamma}{\beta} & \mbox{, } a_1\leq a_2-\frac{1}{\beta}, \\
		\left(\frac{1+\beta a_1-\alpha a_2}{\beta-\alpha}\right)^2+\gamma\frac{1-\alpha(a_2-a_1)}{\beta-\alpha} & \mbox{, } a_2-\frac{1}{\beta}\leq a_1\leq a_2.
	\end{array}
\right.
\]
The unconstrained minimum of the quadratic cost in the second region is then
\[
Y=-\frac{2\beta(1-\alpha a_2)+\gamma\alpha(\beta-\alpha)}{2\beta^2}.
\]
This yields the best response function for user 1,
\begin{equation}\label{eq:br_user1}
b_1(a_2)=\left\{
	\begin{array}{ll}
		-\frac{1}{\beta} & \mbox{, } a_2\geq 0, \\
		 y^*(a_2) & \mbox{, } a_2<0,
	\end{array}
\right.
\end{equation}
where
\[
y^*(a):=\left\{
	\begin{array}{ll}
		a & \mbox{, } a\leq Y, \\
		Y & \mbox{, } Y<a<Y+\frac{1}{\beta}, \\
		a-\frac{1}{\beta} & \mbox{, } a\geq Y+\frac{1}{\beta}.
	\end{array}
\right. 
\]

A CNE in this example is any pair of arrival times $a_1\leq a_2$ such that $a_2\in b_2(a_1)$ and $a_1\in b_1(a_2)$. Note that the equilibrium path of an SPNE is not necessarily a CNE, as user $1$ now minimizes his cost with respect to a single action of user $2$ which does not change with his own action. Using similar arguments to those used in the proof of Proposition \ref{prop:N2_SPNE} we can obtain the following characterisation of the CNE points.

\begin{proposition}\label{prop:N2_CNE}
For $N=2$ and $\mathbf{d}^*=(0,0)$ the CNE are characterised as follows.
\begin{enumerate}
\item If $\gamma=0$ then, 
\begin{enumerate}
\item if $\alpha<\frac{\beta}{2}$ the unique CNE is $\left(-\frac{1}{\beta-\alpha},-\frac{1}{\beta-\alpha}\right)$,
\item if $\alpha\geq\frac{\beta}{2}$ then there are two CNE: $\left(-\frac{1}{\beta-\alpha},-\frac{1}{\beta-\alpha}\right)$ and $\left(-\frac{1}{\beta},0\right)$.
\end{enumerate}
\item If $\gamma>0$ then,
\begin{enumerate}
\item if $\alpha<\frac{\beta}{2}$ then,
\begin{enumerate}
\item if $\gamma\leq \frac{\beta-2\alpha}{\alpha(\beta-\alpha)}$ then there is a unique CNE:
\[
a_1=-\frac{1}{\beta-\alpha}-\frac{\gamma\alpha\left(\beta^2-4\alpha\beta\right)}{2\beta^2(\beta-2\alpha)}, \quad
a_2=-\frac{1}{\beta-\alpha}+\frac{\gamma\alpha\left(\beta+2\alpha\right)}{2\beta(\beta-2\alpha)},
\]
\item if $\gamma> \frac{\beta-2\alpha}{\alpha(\beta-\alpha)}$ then any point $(a_1,a_2)$ satisfying 
\[
a_1=a_2-\frac{1}{\beta},\quad a_2\in\left[-\min\left\lbrace\frac{\gamma\alpha}{2\beta},\frac{1}{\beta}\right\rbrace,-\max\left\lbrace\frac{1}{\beta}-\frac{\gamma\alpha}{2(\beta-2\alpha)},0\right\rbrace\right]
\]
is a CNE,
\end{enumerate} 
\item if $\alpha\geq \frac{\beta}{2}$ then any point $(a_1,a_2)$ satisfying 
\[
a_1=a_2-\frac{1}{\beta}, \quad a_2\in\left[-\min\left\lbrace\frac{\gamma\alpha}{2\beta},\frac{1}{\beta}\right\rbrace,0\right],
\]
is a CNE.
\end{enumerate}
\end{enumerate}
\end{proposition}
\begin{proof}
Clearly, in equilibrium $a_1\leq-\frac{1}{\beta}$ and $-\frac{1}{\beta-\alpha}\leq a_2\leq 0$, otherwise trivial reductions can be made to at least one of the cost functions. Plugging $Y$ in the best response function \eqref{eq:br_user1} for the relevant range yields
\[
b_1(a_2)=\left\{
	\begin{array}{ll}
	\frac{\alpha}{\beta}a_2-\frac{1}{\beta}-\frac{\gamma\alpha(\beta-\alpha)}{2\beta^2} & \mbox{, } -\frac{1}{\beta-\alpha	}\leq a_2\leq -\frac{\gamma\alpha}{2\beta}, \\
	a_2-\frac{1}{\beta} & \mbox{, } -\frac{\gamma\alpha}{2\beta}< a_2\leq 0.
	\end{array}
\right.
\]
The best response function \eqref{eq:br_user2} of user $2$ depends on the values of $\gamma$, $\alpha$ and $\beta$ in a slightly more elaborate manner. Specifically, if $\alpha\geq\frac{\beta}{2}$ then
\[
b_2(a_1)=\left\{
	\begin{array}{ll}
		-\frac{1}{\beta} & \mbox{, } a_1<-\frac{1}{\beta-\alpha},\gamma=0, \\
	-\frac{1}{\beta} \ \mathrm{or} \ -\frac{1+\alpha a_1}{\beta-2\alpha} \ & \mbox{, } -\frac{1}{\beta-\alpha}\leq a_1< -\frac{2}{\beta},\gamma=0, \\
		a_1+\frac{1}{\beta} \ & \mbox{, } -\frac{2}{\beta} \leq a_1< -\frac{1}{\beta},\gamma=0, \\
	-\frac{1}{\beta} & \mbox{, } a_1\leq-\frac{2}{\beta},\gamma>0, \\
	a_1+\frac{1}{\beta} & \mbox{, } a_1>-\frac{2}{\beta},\gamma>0,
	\end{array}
\right.
\]
and if $\alpha<\frac{\beta}{2}$ then
\[
b_2(a_1)=\left\{
	\begin{array}{ll}
	-\frac{1}{\beta} & \mbox{, } a_1\leq-\frac{2}{\beta}, \\
	a_1+\frac{1}{\beta} & \mbox{, }-\frac{2}{\beta}<a_1\leq -\frac{2}{\beta}+\frac{\gamma\alpha}{2(\beta-2\alpha)} \\
		-\frac{\alpha}{\beta-2\alpha}a_1-\frac{1}{\beta-2\alpha}+\frac{\gamma\alpha(\beta-\alpha)}{2(\beta-2\alpha)^2} & \mbox{, } -\frac{2}{\beta}+\frac{\gamma\alpha}{2(\beta-2\alpha)}<a_1\leq -\frac{1}{\beta-\alpha}+\frac{\gamma\alpha}{2(\beta-2\alpha)}, \\
			a_1 & \mbox{, } -\frac{1}{\beta-\alpha}+\frac{\gamma\alpha}{2(\beta-2\alpha)}<a_1\leq-\frac{1}{\beta}.
	\end{array}
\right.
\]
Note that some the regions in the latter function may be empty, depending on the value of $\gamma$.
\begin{enumerate}
\item If $\gamma=0$ then $a_1=a_2=-\frac{1}{\beta-\alpha}$ is clearly an equilibrium because both users incur the minimal cost of zero. Furthermore, in this case,
\[
b_1(a_2)=\frac{\alpha}{\beta}a_2-\frac{1}{\beta},\quad -\frac{1}{\beta-\alpha}\leq a_2\leq 0,
\]
which further implies that $-\frac{1}{\beta-\alpha}\leq b_1(a_2)\leq-\frac{1}{\beta}$. 
\begin{enumerate}
\item If $\alpha<\frac{\beta}{2}$ then $\beta-2\alpha>0$ and $-\frac{2}{\beta}>-\frac{1}{\beta-\alpha}$, and thus $b_2(b_1(a_2))=b_1(a_2)$ for $-\frac{1}{\beta-\alpha}\leq b_1(a_2)\leq-\frac{1}{\beta}$. The only solution to 
\[
a_2=b_2(b_1(a_2))=b_1(a_2)=\frac{\alpha}{\beta}a_2-\frac{1}{\beta},
\]
is then given by $a_2=-\frac{1}{\beta-\alpha}$. 
\item If $\alpha\geq\frac{\beta}{2}$ then there are two solutions to $a_2=b_2(b_1(a_2))$: the first is in the second region of $b_2(a_1)$,
\[
a_2=-\frac{1+\alpha \left(\frac{\alpha}{\beta}a_2-\frac{1}{\beta}\right)}{\beta-2\alpha},
\]
yielding $a_1=a_2=-\frac{1}{\beta-\alpha}$, and the second is $a_2=a_1+\frac{1}{\beta}=0$ in the third region of $b_2(a_1)$.
\end{enumerate}
\item If $\gamma>0$ then $a_1=a_2$ cannot be an equilibrium. This is because $b_1(a_2)=a_2$ only in case $a_2=-\frac{1}{\beta-\alpha}$, which is never a best response for user $2$ if $\gamma>0$. Therefore, there are two possible types of equilibrium: $a_2=a_1+\frac{1}{\beta}$ (and $a_1=a_2-\frac{1}{\beta}$), or the simultaneous solution to the equations $a_1=Y$ and $a_2=X$.
\begin{enumerate}
\item If $\alpha<\frac{\beta}{2}$ then,
\begin{enumerate}
\item If $\gamma\leq \frac{\beta-2\alpha}{\alpha(\beta-\alpha)}$ then $a_2=a_1+\frac{1}{\beta}$ is a best response to $a_1=a_2-\frac{1}{\beta}$ if
\[
-\frac{2}{\beta}< a_1 \leq -\frac{2}{\beta}+\frac{\gamma\alpha}{2(\beta-2\alpha)} \ \Leftrightarrow \ -\frac{1}{\beta}< a_2 \leq -\frac{1}{\beta}+\frac{\gamma\alpha}{2(\beta-2\alpha)},
\]
but then $a_2\leq -\frac{\gamma\alpha}{2\beta}$ (from $\gamma\leq \frac{\beta-2\alpha}{\alpha(\beta-\alpha)}$), which implies that $a_1=a_2-\frac{1}{\beta}$ is not a best response. Hence the unique CNE is given by the interior point $(Y,X)$ as stated in the proposition.
\item If $\gamma> \frac{\beta-2\alpha}{\alpha(\beta-\alpha)}$ then there is no CNE in the interior. All pairs satisfying the condition in the proposition statement are such that best response of user 2 to $a_1=a_2-\frac{1}{\beta}$ is $a_2=a_1+\frac{1}{\beta}$, and vice versa. The interval for $a_2$ is obtained by considering the best response regions of $b_1(a_2)$ and $b_2(a_1)$.
\end{enumerate} 
\item If $\alpha\geq\frac{\beta}{2}$ then $b_2(a_1)$ never has an interior solution, i.e. given by $X$. Therefore all equilibrium are of the form $a_2=a_1+\frac{1}{\beta}$, and are given by the range specified in the proposition statement.\qed
\end{enumerate}
\end{enumerate} 
\end{proof}

\begin{remark}
The multiple equilibria, in fact a continuum of equilibria, arises in cases where interaction is too costly, except for the case $\gamma=0$ which yields a different type of non-unique solution. Specifically, the equilibrium is not unique when there is a steep slowdown in travel speed: $\alpha\geq \frac{\beta}{2}$, or when the travel time penalty $\gamma$ is very high. In all such equilibria the users travel alone at free flow speed.
\end{remark}

We computed the CNE arrival ranges for different parameter values, along with the realised arrival time corresponding to the SPNE. In Figure \ref{fig:arrivals_CNE_SPNE_1} we set $\alpha<\frac{\beta}{2}$, and in this case the CNE are unique and close to the SPE actions for small $\gamma$, whereas for larger values of $\gamma$ there is an expanding range of possible CNE such that both users travel alone. In Figure \ref{fig:arrivals_CNE_SPNE_2} we illustrate an example with $\alpha> \frac{\beta}{2}$. In this case the CNE is never unique and can be quite far from the SPNE actions for larger values of $\gamma$.

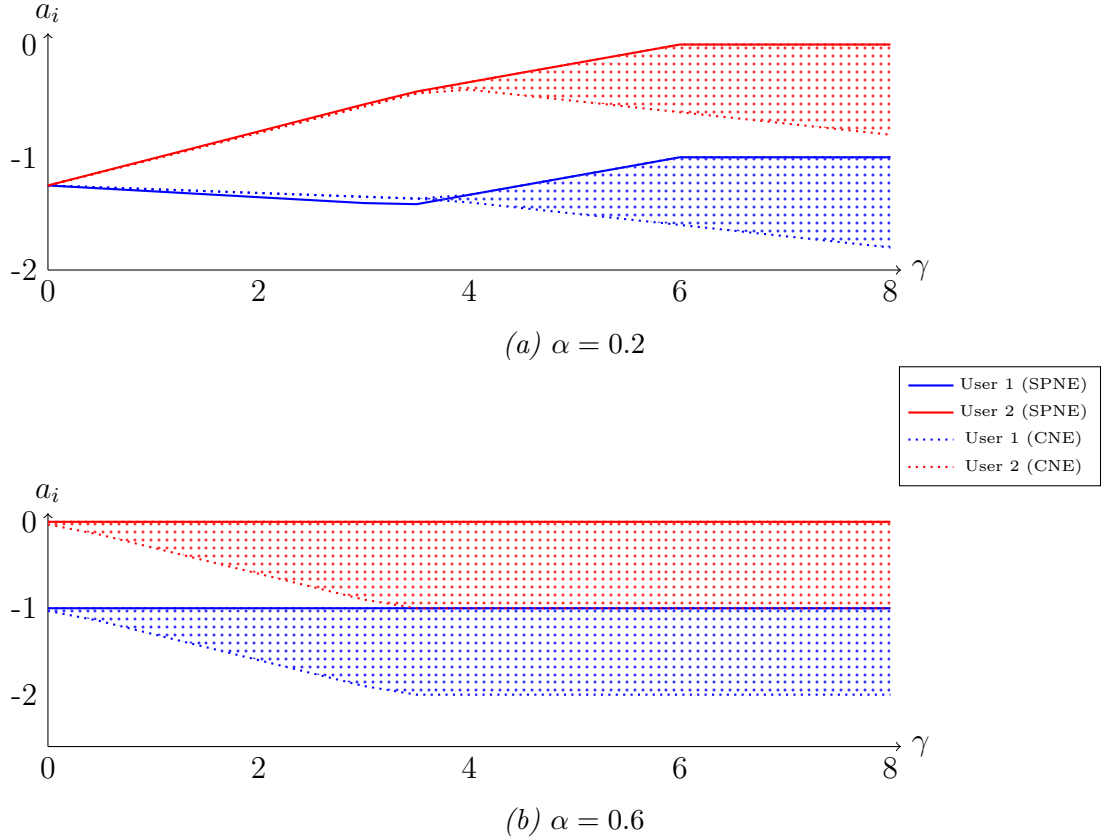
\begin{figure}[H]
\centering
\begin{subfigure}{\textwidth}
\begin{tikzpicture}[xscale=1.4,yscale=1.5]
  \def\xmin{0}
  \def\xmax{8.1}
  \def\ymin{-2}
  \def\ymax{0.1}
    \draw[->] (\xmin,\ymin) -- (\xmax,\ymin) node[right] {$\gamma$} ;
    \draw[->] (0,\ymin) -- (0,\ymax) node[above] {$a_i$} ;
    \foreach \x in {0,2,4,6,8}
    \node at (\x,\ymin) [below] {\x};
    \foreach \y in {-2,-1,0}
    \node at (0,\y) [left] {\y};

	\draw[smooth,blue,thick] (0,-1.25)--	(0.5,-1.276)--	(1,-1.302)--	(1.5,-1.328)--	(2,-1.354)--	(2.5,-1.38)--	(3,-1.406)--	(3.5,-1.416)--	(4,-1.333)--	(4.5,-1.25)--	(5,-1.167)--	(5.5,-1.084)--	(6,-1)--	(6.5,-1)--	(7,-1)--	(7.5,-1)--	(8,-1);
	
	\draw[smooth,red,thick] (0,-1.25)--	(0.5,-1.1302)--	(1,-1.0104)--	(1.5,-0.8907)--	(2,-0.7709)--	(2.5,-0.6511)--	(3,-0.5313)--	(3.5,-0.4169)--	(4,-0.3334)--	(4.5,-0.25)--	(5,-0.167)--	(5.5,-0.084)--	(6,0)--	(6.5,0)--	(7,0)--	(7.5,0)--	(8,0);
	
		\draw[smooth,blue,dotted,thick] (0,-1.25)--	(0.5,-1.266667)--	(1,-1.283333)--	(1.5,-1.3)--	(2,-1.316667)--	(2.5,-1.333333)--	(3,-1.35)--	(3.5,-1.366667)--	(4,-1.4)--	(4.5,-1.45)--	(5,-1.5)--	(5.5,-1.55)--	(6,-1.6)--	(6.5,-1.65)--	(7,-1.7)--	(7.5,-1.75)--	(8,-1.8);
		
		\draw[smooth,blue,dotted,thick] (0,-1.25)--	(0.5,-1.266667)--	(1,-1.283333)--	(1.5,-1.3)--	(2,-1.316667)--	(2.5,-1.333333)--	(3,-1.35)--	(3.5,-1.366667)--	(4,-1.333333)--	(4.5,-1.25)--	(5,-1.166667)--	(5.5,-1.083333)--	(6,-1)--	(6.5,-1)--	(7,-1)--	(7.5,-1)--	(8,-1);
		
		\fill[pattern color=blue!80, opacity=1,pattern=dots] (0,-1.25)--	(0.5,-1.266667)--	(1,-1.283333)--	(1.5,-1.3)--	(2,-1.316667)--	(2.5,-1.333333)--	(3,-1.35)--	(3.5,-1.366667)--	(4,-1.4)--	(4.5,-1.45)--	(5,-1.5)--	(5.5,-1.55)--	(6,-1.6)--	(6.5,-1.65)--	(7,-1.7)--	(7.5,-1.75)--	(8,-1.8)--	(8,-1)--	(7.5,-1)--	(7,-1)--	(6.5,-1)--	(6,-1)--	(5.5,-1.083333)--	(5,-1.166667)--	(4.5,-1.25)--	(4,-1.333333)--	(3.5,-1.366667)--	(3,-1.35)--	(2.5,-1.333333)--	(2,-1.316667)--	(1.5,-1.3)--	(1,-1.283333)--	(0.5,-1.266667)--	(0,-1.25);

	\draw[smooth,red,dotted,thick] (0,-1.25)--	(0.5,-1.1333333)--	(1,-1.0166667)--	(1.5,-0.9)--	(2,-0.7833333)--	(2.5,-0.6666667)--	(3,-0.55)--	(3.5,-0.4333333)--	(4,-0.4)--	(4.5,-0.45)--	(5,-0.5)--	(5.5,-0.55)--	(6,-0.6)--	(6.5,-0.65)--	(7,-0.7)--	(7.5,-0.75)--	(8,-0.8);

	\draw[smooth,red,dotted,thick] (0,-1.25)--	(0.5,-1.13333333)--	(1,-1.01666667)--	(1.5,-0.9)--	(2,-0.78333333)--	(2.5,-0.66666667)--	(3,-0.55)--	(3.5,-0.43333333)--	(4,-0.33333333)--	(4.5,-0.25)--	(5,-0.16666667)--	(5.5,-0.08333333)--	(6,0)--	(6.5,0)--	(7,0)--	(7.5,0)--	(8,0);
	
	\fill[pattern color=red!80, opacity=1,pattern=dots] (0,-1.25)--	(0.5,-1.1333333)--	(1,-1.0166667)--	(1.5,-0.9)--	(2,-0.7833333)--	(2.5,-0.6666667)--	(3,-0.55)--	(3.5,-0.4333333)--	(4,-0.4)--	(4.5,-0.45)--	(5,-0.5)--	(5.5,-0.55)--	(6,-0.6)--	(6.5,-0.65)--	(7,-0.7)--	(7.5,-0.75)--	(8,-0.8)--	(8,0)--	(7.5,0)--	(7,0)--	(6.5,0)--	(6,0)--	(5.5,-0.08333333)--	(5,-0.16666667)--	(4.5,-0.25)--	(4,-0.33333333)--	(3.5,-0.43333333)--	(3,-0.55)--	(2.5,-0.66666667)--	(2,-0.78333333)--	(1.5,-0.9)--	(1,-1.01666667)--	(0.5,-1.13333333)--	(0,-1.25);

\end{tikzpicture}\caption{$\alpha=0.2$}\label{fig:arrivals_CNE_SPNE_1}
\end{subfigure}

\begin{subfigure}{\textwidth}
\begin{tikzpicture}[xscale=1.4,yscale=1.15]
  \def\xmin{0}
  \def\xmax{8.1}
  \def\ymin{-2.6}
  \def\ymax{0.1}
    \draw[->] (\xmin,\ymin) -- (\xmax,\ymin) node[right] {$\gamma$} ;
    \draw[->] (0,\ymin) -- (0,\ymax) node[above] {$a_i$} ;
    \foreach \x in {0,2,4,6,8}
    \node at (\x,\ymin) [below] {\x};
    \foreach \y in {-2,-1,0}
    \node at (0,\y) [left] {\y};

	\draw[smooth,blue,thick] (0,-1)--	(0.5,-1)--	(1,-1)--	(1.5,-1)--	(2,-1)--	(2.5,-1)--	(3,-1)--	(3.5,-1)--	(4,-1)--	(4.5,-1)--	(5,-1)--	(5.5,-1)--	(6,-1)--	(6.5,-1)--	(7,-1)--	(7.5,-1)--	(8,-1);

	\draw[smooth,red,thick] (0,0)--	(0.5,0)--	(1,0)--	(1.5,0)--	(2,0)--	(2.5,0)--	(3,0)--	(3.5,0)--	(4,0)--	(4.5,0)--	(5,0)--	(5.5,0)--	(6,0)--	(6.5,0)--	(7,0)--	(7.5,0)--	(8,0);

		\draw[smooth,blue,dotted,thick] (0,-1.03)--	(0.5,-1.15)--	(1,-1.3)--	(1.5,-1.45)--	(2,-1.6)--	(2.5,-1.75)--	(3,-1.9)--	(3.5,-2)--	(4,-2)--	(4.5,-2)--	(5,-2)--	(5.5,-2)--	(6,-2)--	(6.5,-2)--	(7,-2)--	(7.5,-2)--	(8,-2);

		\draw[smooth,blue,dotted,thick] (0,-1)--	(0.5,-1)--	(1,-1)--	(1.5,-1)--	(2,-1)--	(2.5,-1)--	(3,-1)--	(3.5,-1)--	(4,-1)--	(4.5,-1)--	(5,-1)--	(5.5,-1)--	(6,-1)--	(6.5,-1)--	(7,-1)--	(7.5,-1)--	(8,-1);

		\fill[pattern color=blue!80, opacity=1,pattern=dots] (0,-1.03)--	(0.5,-1.15)--	(1,-1.3)--	(1.5,-1.45)--	(2,-1.6)--	(2.5,-1.75)--	(3,-1.9)--	(3.5,-2)--	(4,-2)--	(4.5,-2)--	(5,-2)--	(5.5,-2)--	(6,-2)--	(6.5,-2)--	(7,-2)--	(7.5,-2)--	(8,-2)--	(8,-1)--	(7.5,-1)--	(7,-1)--	(6.5,-1)--	(6,-1)--	(5.5,-1)--	(5,-1)--	(4.5,-1)--	(4,-1)--	(3.5,-1)--	(3,-1)--	(2.5,-1)--	(2,-1)--	(1.5,-1)--	(1,-1)--	(0.5,-1)--	(0,-1);

	\draw[smooth,red,dotted,thick] (0,-0.03)--	(0.5,-0.15)--	(1,-0.3)--	(1.5,-0.45)--	(2,-0.6)--	(2.5,-0.75)--	(3,-0.9)--	(3.5,-1)--	(4,-1)--	(4.5,-1)--	(5,-1)--	(5.5,-1)--	(6,-1)--	(6.5,-1)--	(7,-1)--	(7.5,-1)--	(8,-1);

	\draw[smooth,red,dotted,thick] (0,0)--	(0.5,0)--	(1,0)--	(1.5,0)--	(2,0)--	(2.5,0)--	(3,0)--	(3.5,0)--	(4,0)--	(4.5,0)--	(5,0)--	(5.5,0)--	(6,0)--	(6.5,0)--	(7,0)--	(7.5,0)--	(8,0);

	\fill[ pattern color=red!80, opacity=1, pattern=dots] (0,-0.03)--	(0.5,-0.15)--	(1,-0.3)--	(1.5,-0.45)--	(2,-0.6)--	(2.5,-0.75)--	(3,-0.9)--	(3.5,-1)--	(4,-1)--	(4.5,-1)--	(5,-1)--	(5.5,-1)--	(6,-1)--	(6.5,-1)--	(7,-1)--	(7.5,-1)--	(8,-1)--	(8,0)--	(7.5,0)--	(7,0)--	(6.5,0)--	(6,0)--	(5.5,0)--	(5,0)--	(4.5,0)--	(4,0)--	(3.5,0)--	(3,0)--	(2.5,0)--	(2,0)--	(1.5,0)--	(1,0)--	(0.5,0)--	(0,0);

	\begin{customlegend}
    [legend entries={User $1$ (SPNE),User $2$ (SPNE),User $1$ (CNE),User $2$ (CNE)},
    legend style={font=\tiny,at={(10,1.8)}}]
    \addlegendimage{blue,thick}
    \addlegendimage{red,thick}
        \addlegendimage{blue,thick,dotted}
    \addlegendimage{red,thick,dotted}
    \end{customlegend}
    
	\end{tikzpicture}\caption{$\alpha=0.6$}\label{fig:arrivals_CNE_SPNE_2}
\end{subfigure}
%
%\centering
%\begin{subfigure}[b]{.2\linewidth}
%\begin{tikzpicture}
%    
%\end{tikzpicture}
%\end{subfigure}
\caption{SPNE and CNE arrival times for $\gamma\in[0,8]$ with $\beta=1$, and $\mathbf{d}^*=(0,0)$. The solid lines are the SPNE actions, and the dotted lines and regions represent the range of possible CNE.}
\label{fig:arrivals_CNE_SPNE}
\end{figure}

In Figure \ref{fig:arrivals_OPT_vs_E} the user arrival times in equilibrium are compared with those specified by the socially optimal profile,
\[
(a_1^o,a_2^o)=\argmin_{(a_1,a_2)}\{c_1(a_1,a_2)+c_2(a_1,a_2)\},
\]
which can be obtained as before by solving two two-dimensional quadratic programs and taking the minimum of both solutions. As expected, the SPNE arrival times are not optimal except for the extreme case of $\gamma=0$. Furthermore, if $\gamma$ is large enough then in all cases users travel on disjoint time intervals. Specifically, according to the SPNE user $1$ arrives at $-1$ and leaves at the desired time of $0$, incurring just the travel time cost of $\gamma$, while user $2$ arrives at $0$ and leaves at $1$, thus incurring a tardiness cost of $1$, and a total cost of $\gamma+1$. In the socially optimal profile, still for large values of $\gamma$, both users incur the same cost of $\gamma+0.25$ by arriving at $-1.5$ and $-0.5$, respectively. Note that it can easily be verified that in the two-user game both users incur the same cost under the socially optimal profile, for any value of $\gamma$. The sequential form of the game gives user $1$ the advantage of first action which allows him to travel alone and arrive at his desired time, when $\gamma$ is large enough. The socially optimal profile ``forces'' user $1$ to share the cost and arrive earlier than he would selfishly choose. The CNE displays very similar behaviour to the SPNE when $\gamma$ is small but as it gets bigger there are multiple equilibria that include both the socially optimal and SPNE solutions.

\begin{figure}[H]
\begin{tikzpicture}[xscale=1.3,yscale=1.6]
  \def\xmin{0}
  \def\xmax{8.1}
  \def\ymin{-2}
  \def\ymax{0.1}
    \draw[->] (\xmin,\ymin) -- (\xmax,\ymin) node[right] {$\gamma$} ;
    \draw[->] (0,\ymin) -- (0,\ymax) node[above] {$a_i$} ;
    \foreach \x in {0,2,4,6,8}
    \node at (\x,\ymin) [below] {\x};
    \foreach \y in {-2,-1.5,-1,-0.5,0}
    \node at (0,\y) [left] {\y};

	\draw[smooth,blue,thick] (0,-1.25)--	(0.5,-1.276)--	(1,-1.302)--	(1.5,-1.328)--	(2,-1.354)--	(2.5,-1.38)--	(3,-1.406)--	(3.5,-1.416)--	(4,-1.333)--	(4.5,-1.25)--	(5,-1.167)--	(5.5,-1.084)--	(6,-1)--	(6.5,-1)--	(7,-1)--	(7.5,-1)--	(8,-1);
	
	\draw[smooth,red,thick] (0,-1.25)--	(0.5,-1.1302)--	(1,-1.0104)--	(1.5,-0.8907)--	(2,-0.7709)--	(2.5,-0.6511)--	(3,-0.5313)--	(3.5,-0.4169)--	(4,-0.3334)--	(4.5,-0.25)--	(5,-0.167)--	(5.5,-0.084)--	(6,0)--	(6.5,0)--	(7,0)--	(7.5,0)--	(8,0);
	
		\draw[smooth,blue,dotted,thick] (0,-1.25)--	(0.5,-1.266667)--	(1,-1.283333)--	(1.5,-1.3)--	(2,-1.316667)--	(2.5,-1.333333)--	(3,-1.35)--	(3.5,-1.366667)--	(4,-1.4)--	(4.5,-1.45)--	(5,-1.5)--	(5.5,-1.55)--	(6,-1.6)--	(6.5,-1.65)--	(7,-1.7)--	(7.5,-1.75)--	(8,-1.8);;
		
		\draw[smooth,blue,dotted,thick] (0,-1.25)--	(0.5,-1.266667)--	(1,-1.283333)--	(1.5,-1.3)--	(2,-1.316667)--	(2.5,-1.333333)--	(3,-1.35)--	(3.5,-1.366667)--	(4,-1.333333)--	(4.5,-1.25)--	(5,-1.166667)--	(5.5,-1.083333)--	(6,-1)--	(6.5,-1)--	(7,-1)--	(7.5,-1)--	(8,-1);
		
		\fill[pattern color=blue!80, opacity=1,pattern=dots] (0,-1.25)--	(0.5,-1.266667)--	(1,-1.283333)--	(1.5,-1.3)--	(2,-1.316667)--	(2.5,-1.333333)--	(3,-1.35)--	(3.5,-1.366667)--	(4,-1.4)--	(4.5,-1.45)--	(5,-1.5)--	(5.5,-1.55)--	(6,-1.6)--	(6.5,-1.65)--	(7,-1.7)--	(7.5,-1.75)--	(8,-1.8)--	(8,-1)--	(7.5,-1)--	(7,-1)--	(6.5,-1)--	(6,-1)--	(5.5,-1.083333)--	(5,-1.166667)--	(4.5,-1.25)--	(4,-1.333333)--	(3.5,-1.366667)--	(3,-1.35)--	(2.5,-1.333333)--	(2,-1.316667)--	(1.5,-1.3)--	(1,-1.283333)--	(0.5,-1.266667)--	(0,-1.25);

	\draw[smooth,red,dotted,thick] (0,-1.25)--	(0.5,-1.1333333)--	(1,-1.0166667)--	(1.5,-0.9)--	(2,-0.7833333)--	(2.5,-0.6666667)--	(3,-0.55)--	(3.5,-0.4333333)--	(4,-0.4)--	(4.5,-0.45)--	(5,-0.5)--	(5.5,-0.55)--	(6,-0.6)--	(6.5,-0.65)--	(7,-0.7)--	(7.5,-0.75)--	(8,-0.8);

	\draw[smooth,red,dotted,thick] (0,-1.25)--	(0.5,-1.13333333)--	(1,-1.01666667)--	(1.5,-0.9)--	(2,-0.78333333)--	(2.5,-0.66666667)--	(3,-0.55)--	(3.5,-0.43333333)--	(4,-0.33333333)--	(4.5,-0.25)--	(5,-0.16666667)--	(5.5,-0.08333333)--	(6,0)--	(6.5,0)--	(7,0)--	(7.5,0)--	(8,0);
	
	\fill[pattern color=red!80, opacity=1,pattern=dots] (0,-1.25)--	(0.5,-1.1333333)--	(1,-1.0166667)--	(1.5,-0.9)--	(2,-0.7833333)--	(2.5,-0.6666667)--	(3,-0.55)--	(3.5,-0.4333333)--	(4,-0.4)--	(4.5,-0.45)--	(5,-0.5)--	(5.5,-0.55)--	(6,-0.6)--	(6.5,-0.65)--	(7,-0.7)--	(7.5,-0.75)--	(8,-0.8)--	(8,0)--	(7.5,0)--	(7,0)--	(6.5,0)--	(6,0)--	(5.5,-0.08333333)--	(5,-0.16666667)--	(4.5,-0.25)--	(4,-0.33333333)--	(3.5,-0.43333333)--	(3,-0.55)--	(2.5,-0.66666667)--	(2,-0.78333333)--	(1.5,-0.9)--	(1,-1.01666667)--	(0.5,-1.13333333)--	(0,-1.25);
	
			\draw[smooth,blue,densely dashed, thick] (0,-1.25)--	(0.5,-1.3125)--	(1,-1.375)--	(1.5,-1.4375)--	(2,-1.4999)--	(2.5,-1.4999)--	(3,-1.4999)--	(3.5,-1.4999)--	(4,-1.4999)--	(4.5,-1.5)--	(5,-1.5)--	(5.5,-1.5)--	(6,-1.5)--	(6.5,-1.5)--	(7,-1.5)--	(7.5,-1.5)--	(8,-1.5);
	
	\draw[smooth,red,densely dashed, thick] (0,-1.25)--	(0.5,-1.0625)--	(1,-0.875)--	(1.5,-0.6875)--	(2,-0.5)--	(2.5,-0.5)--	(3,-0.5)--	(3.5,-0.5)--	(4,-0.5)--	(4.5,-0.4999)--	(5,-0.4999)--	(5.5,-0.4999)--	(6,-0.4999)--	(6.5,-0.4999)--	(7,-0.4999)--	(7.5,-0.4999)--	(8,-0.4999);

    \begin{customlegend}
    [legend entries={User $2$ (SPNE),User $2$ (CNE),User $2$ (OPT),User $1$ (SPNE),User $1$ (CNE),User $1$ (OPT)},
    legend style={font=\tiny,at={(11,-0.3)}}]
    \addlegendimage{red}    
    \addlegendimage{red,thick,dotted}    
    \addlegendimage{red,thick,densely dashed}
    \addlegendimage{blue}
    \addlegendimage{blue,thick,dotted}
    \addlegendimage{blue,thick,densely dashed}    
    \end{customlegend}
\end{tikzpicture}
\caption{SPNE, CNE and socially optimal arrival times for $\gamma\in[0,8]$ with $\beta=1$, $\alpha=0.2$, and $\mathbf{d}^*=(0,0)$. The solid lines are the SPNE outcomes, the dashed lines are the optimal solutions, and the dotted lines and regions represent the range of possible CNE.}
\label{fig:arrivals_OPT_vs_E}
\end{figure}
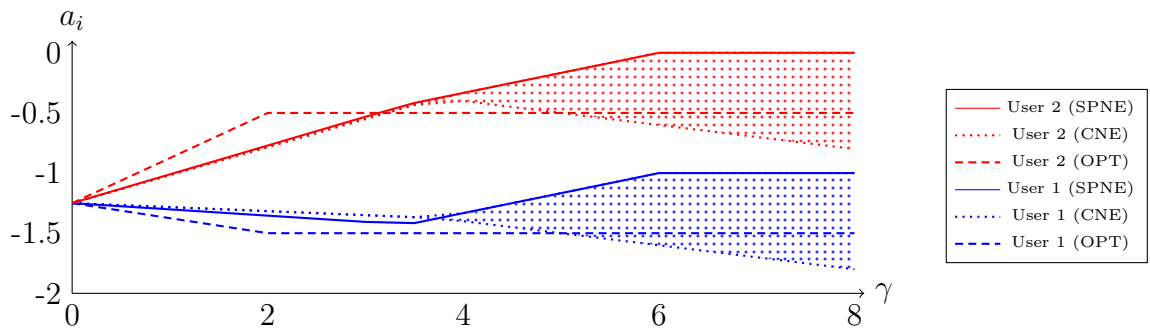

%%%%%%%%%%%%%%%%%%%%%%%%%%%%%
\section{The general arrival game}\label{sec:large_game}
When computing the backward induction or the CNE for the two-user example we had to solve two quadratic programs, one for every order of the pair $(a_2,d_1)$. We showed that already in this simple example the equilibrium was not necessarily unique, in both game formulations, and behaved differently for different parameter values. In order to compute all CNE for any number of users, $N\geq 2$, we must solve a quadratic program for every possible permutation of the departures and arrivals. The backward induction of the SPNE is even more demanding, as at every step $i=1,\ldots,N$, we must take into account the solution of steps $i+1,\ldots,N$, which will yield general (not necessarily quadratic) programs with respect to $a_i$. In Section \ref{sec:traffic_model} we explained that the permutations can be represented by the vectors $\mathbf{k}\in\mathcal{K}$, and made reference to the fact that $|\mathcal{K}|$ grows exponentially (specifically, as the Catalan numbers). Thus, even finding a single equilibrium point, let alone characterising all equilibria, is not an easy task. This leads us to suggest a heuristic method to find CNE points.

%%%%%%%%%%%%%%%%%%%%%%%%%%%%%
\subsection{Ordered best response algorithm}\label{sec:OBRA}
Best response dynamics are often useful for several reasons. Most notably, because points of convergence are Nash equilibrium points. Another advantage is that they are constructive and can provide a feasible evolution of the behaviour of users leading to an equilibrium solution. In many  types of games best response dynamics are known to converge to a unique Nash equilibrium for any initial point. This is the case for example in potential games \cite{MS1996}, and sub-modular games \cite{T1979}. Unfortunately, our game is not an instance of any of these games. We already know that the equilibrium is not necessarily unique. Furthermore, the non-convex and non-smooth form of the cost function suggests there is not much hope for a constructive iterative procedure that is guaranteed to find an equilibrium point in polynomial time. Nevertheless, if we optimistically apply the best response procedure and it does indeed converge for a given initial point then the output is a CNE.

We formally define the best-response algorithm starting at any ordered arrival profile $\mathbf{a}_0$, denoted by $\mathrm{BR}(\mathbf{a}_0)$, as follows.

\paragraph{Best response algorithm - $\mathrm{BR}(\mathbf{a}_0)$}
\begin{enumerate}
\item[(1)] Initiate $\mathbf{\hat{a}}:=\mathbf{a}_0$.
\item[(2)] For $i=1,\ldots,N$, let $\mathcal{B}_i=\mathcal{BR}(a_1,\ldots,a_{i-1},\hat{a}_{i+1},\ldots,\hat{a}_N)$ and compute
\begin{equation}
a_i=\left\{
	\begin{array}{ll}
		\hat{a}_i & \mbox{, } \hat{a}_i\in \mathcal{B}_i, \\
		\argmin\mathcal{B}_i & \mbox{, } \hat{a}_i\notin\mathcal{B}_i.
	\end{array}
\right.
\end{equation}
\item[(3)] If $\hat{a}_i=a_i,\ \forall i=1,\ldots,N$ then stop, otherwise set $\hat{a}_i=a_i$ and go back to step (2).
\end{enumerate}

\begin{remark}
The rule of updating the arrival time of a single user when the minimizer is not unique is defined in a way that guarantees changes will only occur when there is a strictly positive decrease in cost for the user. Specifically if $\hat{a}_i$ is in the best response group then we make no change, and otherwise we choose the smallest element, $\argmin\mathcal{B}_i$. This avoids the algorithm getting stuck in loops caused by multiple minimizers. Examples of such loops can easily be constructed. Note that this does not rule out other kinds of loops. We have not been able to construct other types of loops, nor have we been able to prove they don't exist.
\end{remark}

\begin{remark}
The algorithm may step outside of the action space because the optimization for user $i$ is on the interval $[a_{i-1},\infty)$, regardless of the arrival times of user $j=i+1,\ldots,N$. However, after a full iteration on all users the arrival profile will always be ordered. Furthermore, we argue that this is more natural than limiting the search to the interval $[a_{i-1},a_{i+1}]$, as the game is sequential and in the actual course of play the users have to react to the users preceding them. Note also that if we restrict the order during the iteration then a point of convergence is not necessarily an equilibrium.
\end{remark}

In \cite{RN2015} a method for optimizing the sum of total costs with respect to a single coordinate, while keeping all others fixed, was presented. We will now detail how step (2) of the BR algorithm can be carried out using a modification of the aforementioned method. Let $\mathbf{a}$ be an ordered arrival time vector, and suppose we now want to minimize the cost function $c_i(a_i,\mathbf{a}_{-i})$ by changing only $a_i$, i.e., find the set
\[
\argmin_{a\geq a_{i-1}}\{(d_i(a,\mathbf{a}_{-i})-d_i^*)^2 \,+ \gamma \,  (d_i(a,\mathbf{a}_{-i})-a_i)\}.
\] 
Recall that the model dynamics yield a piecewise-quadratic behaviour of the cost function with respect to any single coordinate, as was illustrated by the solid line in Figure \ref{fig:cost_arrival}. Therefore, the minimizer can be found by computing the minimal value of a single variable quadratic function on every continuous segment. Every segment corresponds to an order of all arrival and departure times (indexed by $\mathbf{k}\in\mathcal{K}$) and the coefficients of the piecewise-linear term $d_i(a,\mathbf{a}_{-i})$ are given by \eqref{eq:d_recursive}. The number of possible segments is given by the number of possible order changes. In Lemma 7 of \cite{RN2015} an upper bound of $\mathcal{O}(N^3)$ was given for the number of possible segments. This upper bound was obtained using worst-case analysis, and the number of computations was typically much less in practice.

We have shown in Proposition \ref{prop:N2_CNE} that there may exist multiple equilibria. Moreover, the algorithm may not converge at all if there exist loops in the best-response dynamics. To address these problems we suggest running the algorithm on a set of initial profiles, denoted by $\mathcal{A}_0$, which may be generated randomly or according to some specified structure. In addition, we define a maximum iteration parameter that aborts the algorithm if it does not converge within the predefined number of iterations. We denote the repeated algorithm by $\mathrm{RBR}(\mathcal{A}_0)$. This will also allow us to numerically analyse the price of anarchy which is defined as the ratio between the cost corresponding to the worst CNE and that of the socially optimal arrival profile. Formally,
\[
\mathrm{PoA}:=\frac{\max_{\mathbf{a}^e\in\mathcal{A}^e}\sum_{i\in\mathcal{N}}c_i(\mathbf{a}^e)}{\sum_{i\in\mathcal{N}}c_i(\mathbf{a}^o)},
\] 
where $\mathcal{A}^e$ is the set of CNE and $\mathbf{a}^o$ is the socially optimal arrival profile. In the numerical analysis we substitute $\mathcal{A}^e$ with the output of $\mathrm{RBR}(\mathcal{A}_0)$.

%%%%%%%%%%%%%%%%%%%%%%%%%%%%%
\section{Numerical analysis}\label{sec:numerical}
Throughout this section we will compute equilibrium results using the $\mathrm{BR}$ and $\mathrm{RBR}$ algorithms and compare them with the socially optimal solution obtained using the algorithms presented in \cite{RN2015}. When possible ($N\leq 15$), the optimal solution is exact and otherwise it is a result of a polynomial-time search method that is guaranteed to converge to a local minimum. We use superscript notations of ``$o$'' and ``$e$'' to denote socially optimal and CNE, respectively.

In table \ref{tbl:numerical_RBR} we summarize the convergence statistics of 100 instances of the $RBR$ with randomized initial profiles, for different user population sizes. 

\begin{table}[H]
\centering
\scalebox{0.65}{
\begin{tabular}{|c|c|c|c|c|c|} \hline $N$ & Convergence $\%$ & Mean $\#$ of Iterations & Min $\#$ of Iterations & Max $\#$ of Iterations & $\#$ of distinct CNE  \\ \hline
$20$ & $100\%$  & $6.64$ & $5$ & $8$ & $2$ \\ \hline
$50$ & $100\%$  & $7.11$ & $6$ & $8$ & $2$ \\ \hline
$80$ & $100\%$  & $7.44$ & $7$ & $9$ & $5$ \\ \hline
\end{tabular}}
\caption{Summary of 100 $RBR$ instances for three examples with different population sizes, and parameters $\beta=1$, $\alpha=0.7\frac{\beta}{N}$, $\gamma=\frac{1}{N}$, and $\mathbf{d}^*$ equals the $N$ quantiles of a normal distribution with mean 0 and variance 1. The random initial points where selected as follows: 50 vectors were sampled from as the order statistic of an iid uniform distribution on $(-1,1)$, and 50 vectors were sampled from a normal distribution with mean $\mathbf{d}^*-\frac{1}{\beta}$ and variance between $0$ and $2$.}
\label{tbl:numerical_RBR}
\end{table}

All instances converged to equilibrium points within a small number of iterations. Multiple equilibria were found in all three examples, although the different CNE do not seem significantly different from each other. For example in the example of $N=50$ the maximum absolute distance between any two equilibrium arrival schedules was
\[
\max_{\mathbf{\hat{a}},\mathbf{\bar{a}}\in\mathcal{A}^e}\sum_{i=1}^N|\hat{a}_i-\bar{a}_i|=1.16,
\]
and the maximum distance between the arrival times of the same user in different equilibria was
\[
\max_{\mathbf{\hat{a}},\mathbf{\bar{a}}\in\mathcal{A}^e}\max_{i=1,\ldots,N}|\hat{a}_i-\bar{a}_i|=0.09.
\]
To put these numbers in perspective, the free flow travel time in this example was $1$, the slowest travel time in a full system was $\sim 4.3$, and the time from the first arrival to the last departure in equilibrium was $\sim 5$. Therefore, a deviation of at most $0.1$ between arrival times in all schedules suggests that all CNE schedules were effectively the same, and the small differences were perhaps due to the very non-smooth behaviour of the cost function. Recall that in Proposition \ref{prop:N2_CNE} it was shown that for $N=2$ infinite equilibria may exist on a continuous interval, hence it is not surprising if this occurs in the general game as well for some small interval.

Despite the positive convergence results demonstrated in the above examples, it should be noted that the scaling of the parameters is important. We ran the $N=50$ example twice more, with a single parameter change each time, first we set $\alpha=0.9\frac{\beta}{N}$, and then we set $\mathbf{d}^*$ as the $N$ quantiles of a normal distribution with mean 0 and variance 0.1 ($\alpha$ was as in the original example). In both cases the algorithm did not converge on any of the 100 attempts after 100 iterations. In both of these examples the negative effects users incur on one another are very high. The lack of convergence may indicate that the algorithm suffers from numerical precision difficulties when users want to arrive at virtually the same time, with respect to the congested travel speeds. If interaction costs are very high (not scaled) it is likely that there will be almost no overlap between users in equilibrium, and in this case other approximation methods, perhaps of a combinatorial nature, may be more suitable than the Best-Response algorithm presented here. 

Figure \ref{fig:ds_de_do} shows the CNE and socially optimal departure times compared with the desired departure times for a specific set of parameters. In equilibrium users depart from the system closer to their desired departure time than under the socially optimal policy. This entails higher congestion costs and overall travel times. Specifically, in this example the average cost for deviation from the desired time is $3.6$ times higher under the socially optimal policy, but the average travel time costs are $3.8$ times lower than in equilibrium. The individual user travel times can be seen in Figure \ref{fig:travel_times}. The price of anarchy in this example is $1.07$. 

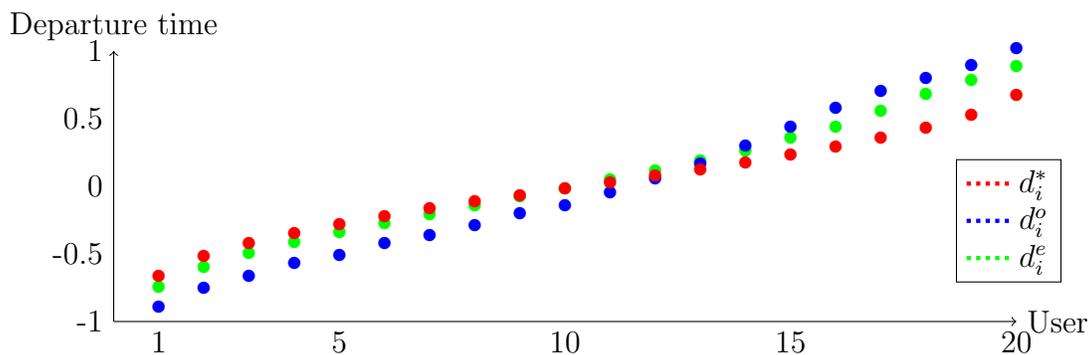
\begin{figure}[H]
\centering
\begin{tikzpicture}[xscale=0.6,yscale=1.8]
  \def\xmin{0}
  \def\xmax{20}
  \def\ymin{-1}
  \def\ymax{1}
    \draw[->] (\xmin,\ymin) -- (\xmax,\ymin) node[right] {User} ;
    \draw[->] (\xmin,\ymin) -- (\xmin,\ymax) node[above] {Departure time} ;
    \foreach \x in {1,5,10,15,20}
    \node at (\x,\ymin) [below] {\x};
    \foreach \y in {-1,-0.5,0,0.5,1}
    \node at (0,\y) [left] {\y};
    
    \foreach \Point in {(1,-0.752),	(2,-0.606),	(3,-0.502),	(4,-0.417),	(5,-0.348),	(6,-0.279),	(7,-0.211),	(8,-0.145),	(9,-0.08),	(10,-0.023),	(11,0.041),	(12,0.112),	(13,0.183),	(14,0.258),	(15,0.353),	(16,0.435),	(17,0.55),	(18,0.676),	(19,0.781),	(20,0.881)}	
    {\node[green] at \Point {\textbullet};}

    \foreach \Point in {(1,-0.89965704),	(2,-0.7563426),	(3,-0.66907036),	(4,-0.57254261),	(5,-0.51251957),	(6,-0.42927757),	(7,-0.36726784),	(8,-0.29291126),	(9,-0.20816913),	(10,-0.14739121),	(11,-0.05317523),	(12,0.04931436),	(13,0.16271619),	(14,0.29040115),	(15,0.43684372),	(16,0.57656515),	(17,0.6952148),	(18,0.79270345),	(19,0.89216265),	(20,1.01243784)}
    {\node[blue] at \Point {\textbullet};}

    \foreach \Point in {(1,-0.66735648),	(2,-0.52366869),	(3,-0.42702821),	(4,-0.35045714),	(5,-0.28497721),	(6,-0.22637953),	(7,-0.17229092),	(8,-0.12119218),	(9,-0.07200495),	(10,-0.02388684),	(11,0.02388684),	(12,0.07200495),	(13,0.12119218),	(14,0.17229092),	(15,0.22637953),	(16,0.28497721),	(17,0.35045714),	(18,0.42702821),	(19,0.52366869),	(20,0.66735648)}	
    {\node[red] at \Point {\textbullet};}

    \begin{customlegend}
    [legend entries={$d_i^*$, $d_i^o$, $d_i^e$},
    legend style={at={(21,0.2)}}]    
    \addlegendimage{red,dotted,ultra thick}    
    \addlegendimage{blue,dotted,ultra thick} 
    \addlegendimage{green,dotted,ultra thick}    
    \end{customlegend}
\end{tikzpicture}
\caption{Desired, CNE, and socially optimal departure times for $N=20$, $\beta=1$, $\alpha=0.035$, $\gamma=0.35$, and ${\mathbf d}^*$ equals the $20$ quantiles of a normal distribution with mean 0 and variance $0.16$.}
\label{fig:ds_de_do}
\end{figure}

\begin{figure}[H]
\centering
\begin{tikzpicture}[xscale=0.6,yscale=3.5]
  \def\xmin{0}
  \def\xmax{20}
  \def\ymin{1.1}
  \def\ymax{2.1}
    \draw[->] (\xmin,\ymin) -- (\xmax,\ymin) node[right] {User} ;
    \draw[->] (\xmin,\ymin) -- (\xmin,\ymax) node[above] {Travel time} ;
    \foreach \x in {1,5,10,15,20}
    \node at (\x,\ymin) [below] {\x};
    \foreach \y in {1.5,2}
    \node at (0,\y) [left] {\y};
    
    \foreach \Point in {(1,1.651),	(2,1.736),	(3,1.79),	(4,1.831),	(5,1.865),	(6,1.893),	(7,1.914),	(8,1.931),	(9,1.941),	(10,1.949),	(11,1.95),	(12,1.939),	(13,1.924),	(14,1.897),	(15,1.854),	(16,1.796),	(17,1.702),	(18,1.531),	(19,1.421),	(20,1.319)}	
    {\node[green] at \Point {\textbullet};}

    \foreach \Point in {(1,1.456823),	(2,1.518378),	(3,1.56966),	(4,1.598573),	(5,1.630515),	(6,1.649327),	(7,1.671386),	(8,1.682633),	(9,1.695242),	(10,1.699009),	(11,1.699009),	(12,1.686996),	(13,1.658744),	(14,1.607892),	(15,1.524747),	(16,1.440012),	(17,1.379879),	(18,1.316815),	(19,1.252993),	(20,1.186213)}
    {\node[blue] at \Point {\textbullet};}

    \begin{customlegend}
    [legend entries={ $d_i^e-a_i^e$,$d_i^o-a_i^o$},
    legend style={at={(10.5,1.52)}}]   
    \addlegendimage{green,dotted,ultra thick}     
    \addlegendimage{blue,dotted,ultra thick}    
    \end{customlegend}
\end{tikzpicture}
\caption{Individual user travel times of CNE and socially optimal profiles, for $N=20$, $\beta=1$, $\alpha=0.035$, $\gamma=0.35$, and ${\mathbf d}^*$ equals the $20$ quantiles of a normal distribution with mean 0 and variance $0.16$.}
\label{fig:travel_times}
\end{figure}
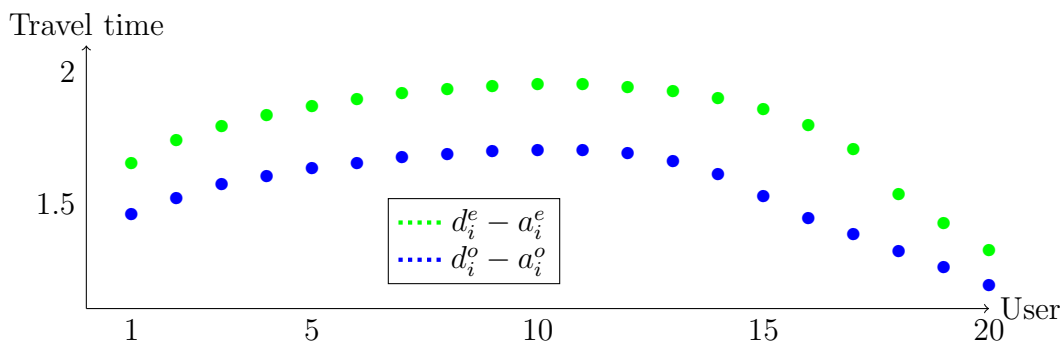

In Figure \ref{fig:co_ce} we illustrate the individual user costs under both schemes. In equilibrium the cost is uni-modal with the users wishing to arrive in the middle of the interval incurring the most cost. The cost under the socially optimal policy appears bi-modal, with the cost more spread out. In Figure \ref{fig:queue_size} we illustrate the smoothed queue size dynamics for both profiles. In equilibrium the system fills up faster, reaches a higher level, and is cleared earlier. This is because arrival times are less spread out in equilibrium. Hence, although users travel faster under the socially optimal policy, the last users will leave later simply because they arrived later than had they been free to choose their arrival time.

\begin{figure}[H]
\centering
\begin{tikzpicture}[xscale=0.6,yscale=17]
  \def\xmin{0}
  \def\xmax{20}
  \def\ymin{0.42}
  \def\ymax{0.6}
    \draw[->] (\xmin,\ymin) -- (\xmax,\ymin) node[right] {User} ;
    \draw[->] (\xmin,\ymin) -- (\xmin,\ymax) node[above] {Cost} ;
    \foreach \x in {1,5,10,15,20}
    \node at (\x,\ymin) [below] {\x};
    \foreach \y in {0.45,0.5,0.55}
    \node at (0,\y) [left] {\y};
    
    \foreach \Point in {(1,0.505),	(2,0.529),	(3,0.545),	(4,0.557),	(5,0.566),	(6,0.573),	(7,0.579),	(8,0.583),	(9,0.586),	(10,0.588),	(11,0.588),	(12,0.587),	(13,0.584),	(14,0.58),	(15,0.574),	(16,0.563),	(17,0.55),	(18,0.53),	(19,0.494),	(20,0.436)}	
    {\node[green] at \Point {\textbullet};}

    \foreach \Point in {(1,0.4990764),	(2,0.5203162),	(3,0.5364464),	(4,0.5391721),	(5,0.5481705),	(6,0.5456932),	(7,0.5472782),	(8,0.5436911),	(9,0.5362799),	(10,0.5343522),	(11,0.5250374),	(12,0.5163025),	(13,0.5097674),	(14,0.508124),	(15,0.5158982),	(16,0.5278622),	(17,0.5360781),	(18,0.5305847),	(19,0.5136489),	(20,0.4767662)}
    {\node[blue] at \Point {\textbullet};}

    \begin{customlegend}
    [legend entries={ $c_i^e$,$c_i^o$},
    legend style={at={(10.5,0.5)}}]   
    \addlegendimage{green,dotted,ultra thick}     
    \addlegendimage{blue,dotted,ultra thick}    
    \end{customlegend}
\end{tikzpicture}
\caption{Individual user costs of CNE and socially optimal profiles, for $N=20$, $\beta=1$, $\alpha=0.035$, $\gamma=0.35$, and ${\mathbf d}^*$ equals the $20$ quantiles of a normal distribution with mean 0 and variance $0.16$.}
\label{fig:co_ce}
\end{figure}
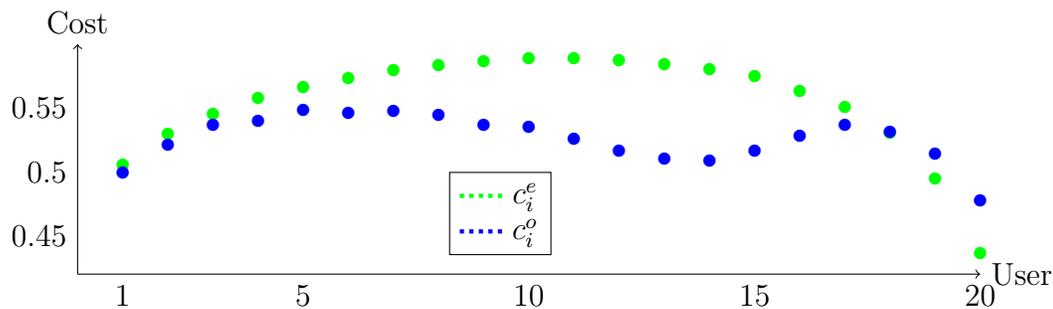

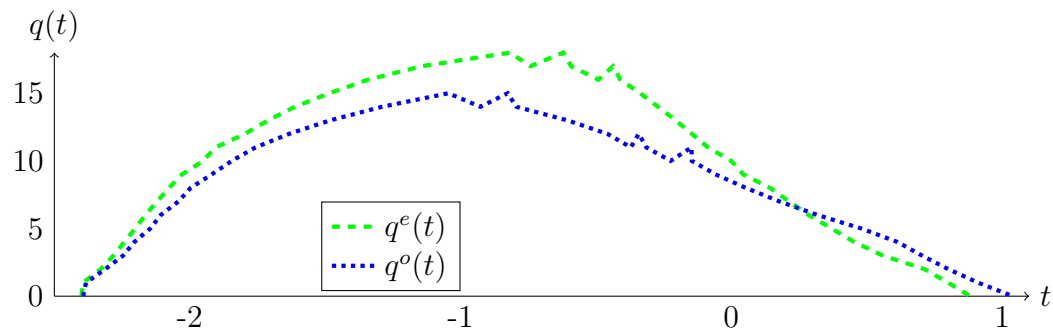
\begin{figure}[H]
\centering
\begin{tikzpicture}[xscale=3.6,yscale=0.18]
  \def\xmin{-2.5}
  \def\xmax{1.1}
  \def\ymin{0}
  \def\ymax{18}
    \draw[->] (\xmin,\ymin) -- (\xmax,\ymin) node[right] {$t$} ;
    \draw[->] (\xmin,\ymin) -- (\xmin,\ymax) node[above] {$q(t)$} ;
    \foreach \x in {-2,-1,0,1}
    \node at (\x,\ymin) [below] {\x};
    \foreach \y in {0,5,10,15}
    \node at (\xmin,\y) [left] {\y};
    
    \draw[smooth,green,dashed, ultra thick] (-2.4,0)--(-2.394,1)--	(-2.329,2)--	(-2.282,3)--	(-2.242,4)--	(-2.202,5)--	(-2.164,6)--	(-2.121,7)--	(-2.077,8)--	(-2.029,9)--	(-1.951,10)--	(-1.905,11)--	(-1.796,12)--	(-1.711,13)--	(-1.612,14)--	(-1.485,15)--	(-1.342,16)--	(-1.143,17)--	(-0.818,18)--	(-0.742,17)--	(-0.621,18)--	(-0.592,17)--	(-0.491,16)--	(-0.437,17)--	(-0.41,16)--	(-0.337,15)--	(-0.271,14)--	(-0.206,13)--	(-0.146,12)--	(-0.087,11)--	(-0.001,10)--	(0.044,9)--	(0.143,8)--	(0.213,7)--	(0.285,6)--	(0.37,5)--	(0.454,4)--	(0.559,3)--	(0.713,2)--	(0.8,1)--	(0.882,0); 
    
    \draw[smooth,blue,dotted, ultra thick] (-2.394,0)--	(-2.383,1)--	(-2.312,2)--	(-2.246,3)--	(-2.205,4)--	(-2.146,5)--	(-2.107,6)--	(-2.042,7)--	(-1.997,8)--	(-1.919,9)--	(-1.845,10)--	(-1.757,11)--	(-1.633,12)--	(-1.481,13)--	(-1.29,14)--	(-1.045,15)--	(-0.927,14)--	(-0.826,15)--	(-0.794,14)--	(-0.606,13)--	(-0.458,12)--	(-0.371,11)--	(-0.342,12)--	(-0.314,11)--	(-0.224,10)--	(-0.147,11)--	(-0.146,10)--	(-0.058,9)--	(0.054,8)--	(0.178,7)--	(0.318,6)--	(0.48,5)--	(0.614,4)--	(0.704,3)--	(0.801,2)--	(0.911,1)--	(1.039,0);

    \begin{customlegend}
    [legend entries={ $q^e(t)$,$q^o(t)$},
    legend style={at={(-1,7)}}]   
    \addlegendimage{green,dashed,ultra thick}     
    \addlegendimage{blue,dotted,ultra thick}    
    \end{customlegend}
\end{tikzpicture}
\caption{Smoothed queue size dynamics of the CNE and socially optimal profiles, for $N=20$, $\beta=1$, $\alpha=0.035$, $\gamma=0.35$, and ${\mathbf d}^*$ equals the $20$ quantiles of a normal distribution with mean 0 and variance $0.16$.}
\label{fig:queue_size}
\end{figure}

It should be highlighted that the patterns exhibited in Figures \ref{fig:ds_de_do}-\ref{fig:queue_size} are robust with respect to the game parameters. In particular, the difference between equilibrium and optimal arrival profiles were in the same direction for all instances tested, with varying effect sizes, and in this sense the examples presented here are a good representation of the general case.

%%%%%%%%%%%%%%%%%%%%%%%%%%%%%
\section{Concluding remarks}\label{sec:conclusion}
In this paper we have presented and analysed a novel ordered arrival time game to a congested processor sharing system, with a heterogeneous discrete user population. For a two-user example we provided explicit equilibrium analysis for two solution concepts, the Subgame Perferct Nash Equilibrium and the Cournot Nash Equilibrium. This example highlighted the required steps for computing equilibrium points, and the resulting difficulties. It was shown that the ordered form of the game played an important role in both solutions. The two solutions display similar outcomes for certain game parameters, namely when interaction is not very costly. It was shown that there may exist multiple CNE, and that for some parameter values both the SPNE path and the socially optimal solution are CNE. For the general population model we presented a heuristic algorithm to compute CNE and analysed it numerically. It was shown that this algorithm is very efficient if the parameters are scaled in a manner that does not allow the interaction costs to be too high. This allowed us to compare the equilibrium arrival process and the resulting congestion process with the corresponding socially optimal solution. The numerical analysis suggests that in equilibrium the travel times and queues are longer, but the users arrive closer to their desired times. This conclusion was supported by testing many population sizes and parameter values.
 
There are several straightforward generalizations of the user cost functions that can be made. One can consider the general form:
\[
c_i(\mathbf{a})=f_i\left((d_i-d_i^*)^+\right)+h_i\left((d_i-d_i^*)^+\right)+g_i(d_i-a_i),
\]
where $f_i$, $g_i$, and $h_i$ are continuous convex functions for every $i\in\mathcal{N}$. This allows for non-quadratic deviation penalties and a different penalty for being late or early. There is no substantial difference in the general model and all of our analysis still holds, with slight technical modifications. In particular, one needs to solve constrained convex programs instead of quadratic, and the differentiation in penalties for being late or early will lead to a larger number of piecewise-convex segments. The choice of the specific cost function in this work was made for the sake of brevity. A network arrival time game, in which users can choose arrival times and travel routes, may also be analysed using our model. The travel dynamics will be more complex when users travel on several subsequent segments with a linear slowdown, but can still be solved recursively using the framework we have provided.

A non-straightforward generalization is allowing the users to have heterogeneous service demand, which translates to different travel distances in the traffic context. If this is the case then there is no longer a reason to order the arrivals according to their desired departure times, which adds a new layer of complexity to the analysis. It seems that there is not much hope for tractable solutions in this case and a different approach, perhaps combining a heuristic on order permutations, is called for. 

An additional open research question arising from this work is whether a fluid approximation, in the sense of \cite{KP2006}, of the game can be solved. In the fluid version the piecewise-linear dynamics will take the form of a set of retarded non-linear differential equations. It is unclear to the authors at this point if the corresponding fluid game can be solved directly or perhaps only by using a mean-field approach with less restrictive dynamics.

\small{
\section*{Acknowledgements}
We thank Yoni Nazarathy for his comments and advice and an anonymous reviewer for his/her most helpful comments. We are grateful to the Australia-Israel Scientific Exchange Foundation (AISEF) for supporting the first author's visit to The Swinburne University of Technology. This work was supported by the Australian Research Council (ARC) Future Fellowships grant FT120100723, and the Israel Science Foundation grant no. 1319/11.}

%%%%%%%%%%%%%%%%%%%%%%%%%%%%%
{\footnotesize\bibliography{C:/Users/Liron/Dropbox/University/Research/Full_Bibliography/BigBib}}

\begin{thebibliography}{10}

\bibitem{ADL1993}
R.~Arnott, A.~de~Palma, and R.~Lindsey.
\newblock A structural model of peak-period congestion: A traffic bottleneck
  with elastic demand.
\newblock {\em American Economic Review}, 83(1):161--79, March 1993.

\bibitem{book_OW2011}
J.~de~Dios~Ort{\'u}zar and L.~G. Willumsen.
\newblock {\em Modelling transport}.
\newblock John Wiley \& Sons, 2011.

\bibitem{GH1983}
A.~Glazer and R.~Hassin.
\newblock {?/M/$\textit{1}$: On the equilibrium distribution of customer
  arrivals}.
\newblock {\em European Journal of Operational Research}, 13(2):146--150, 1983.

\bibitem{HR2015}
M.~Haviv and L.~Ravner.
\newblock Strategic timing of arrivals to a finite queue multi-server loss
  system.
\newblock {\em Queueing Systems}, 81(1):71--96, 2015.

\bibitem{JR2014}
S.~Juneja and T.~Raheja.
\newblock The concert queueing game: Fluid regime with random order service.
\newblock {\em International Game Theory Review}, 17(02):1540012, 2015.

\bibitem{JS2013}
S.~Juneja and N.~Shimkin.
\newblock The concert queueing game: strategic arrivals with waiting and
  tardiness costs.
\newblock {\em Queueing Systems}, 74(4):369--402, 2013.

\bibitem{KP2006}
S.~Kachani and G.~Perakis.
\newblock Fluid dynamics models and their applications in transportation and
  pricing.
\newblock {\em European journal of operational research}, 170(2):496--517,
  2006.

\bibitem{MH1984}
H.~Mahmassani and R.~Herman.
\newblock Dynamic user equilibrium departure time and route choice on idealized
  traffic arterials.
\newblock {\em Transportation Science}, 18(4):362--384, 1984.

\bibitem{MS1996}
D.~Monderer and L.~S. Shapley.
\newblock Potential games.
\newblock {\em Games and Economic Behavior}, 14(1):124 -- 143, 1996.

\bibitem{OR2008}
H.~Otsubo and A.~Rapoport.
\newblock Vickrey’s model of traffic congestion discretized.
\newblock {\em Transportation Research Part B: Methodological},
  42(10):873--889, 2008.

\bibitem{R2014}
L.~Ravner.
\newblock Equilibrium arrival times to a queue with order penalties.
\newblock {\em European Journal of Operational Research}, 239(2):456--468,
  2014.

\bibitem{RN2015}
L.~Ravner and Y.~Nazarathy.
\newblock Scheduling for a processor sharing system with linear slowdown.
\newblock {\em arXiv preprint arXiv:1508.03136}, 2015.

\bibitem{T1979}
D.~M. Topkis.
\newblock Equilibrium points in nonzero-sum n-person submodular games.
\newblock {\em SIAM Journal on Control and Optimization}, 17(6):773--787, 1979.

\bibitem{V1969}
W.~S. Vickrey.
\newblock Congestion theory and transport investment.
\newblock {\em The American Economic Review}, 59(2):251--260, 1969.

\end{thebibliography}

\end{document}